\documentclass[onecolumn]{article}
\usepackage[utf8]{inputenc}

\usepackage[backend=bibtex,style=alphabetic]{biblatex}
\usepackage{LindaMath}
\usepackage{todonotes}
\usepackage{xcolor}
\usepackage{biblatex}
\usepackage{indentfirst}
\usepackage{pgf,tikz,pgfplots, tikz-cd}
\usepackage{enumitem}
\usepackage{subfig}
\usepackage{caption}
\usepackage{setspace}
\usepackage[margin=1.3in]{geometry}
\usepackage{hyperref}
\hypersetup{
    colorlinks=true,
    linkcolor=blue,
    filecolor=red,      
    urlcolor=cyan,
    pdftitle={Overleaf Example},
    pdfpagemode=FullScreen,
    }

\DeclareMathOperator{\Dgm}{Dgm}

\title{A Topological Centrality Measure for Directed Networks}
\author{Fenghuan He\footnote{Commonwealth School, Boston, MA. Email: lhe@commschool.org}}

\begin{document}

\maketitle

\begin{abstract}

Given a directed network $ G $, we are interested in studying the qualitative features of $ G $ which govern how perturbations propagate across $ G $.
Various classical centrality measures have been already developed and proven useful to capture qualitative features and behaviors for undirected networks. 
In this paper, we use topological data analysis (TDA) to adapt measures of centrality to capture both directedness and non-local propagating behaviors in networks. 
We introduce a new metric for computing centrality in directed weighted networks, namely the \emph{quasi-centrality} measure. We compute these metrics on trade networks to illustrate that our measure successfully captures propagating effects in the network and can also be used to identify sources of shocks that can disrupt the topology of directed networks. Moreover, we introduce a method that gives a hierarchical representation of the topological influences of nodes in a directed network. 
\end{abstract} 


\tableofcontents
 
\section{Introduction}
Networks are a useful abstraction for many real-world systems, representing interactions between objects within a system. Network analysis examines relationships among entities, such as persons, organizations, or documents \cite{wright_2015} and has been extensively adopted for modeling systems in various domains with a long history of applications \cite{doi:10.1137/S003614450342480} \cite{Costa_2011},  including neuroscience \cite{şimşek2021geometry}, biology \cite{Leifer_2020}, and social networks \cite{doi:10.1177/016555150202800601}. Networks exist in various forms: weighted or unweighted, directed or undirected. However, many complex systems are more accurately modeled by directed weighted networks, where the relation between two different entities in the system is asymmetric and exists in a range of intensities.  
For example, citation networks, immigration networks, and trade networks are all accurately modeled by directed weighted networks. 

Directed networks are mathematically represented as a directed graph where vertices represent the objects or entities in the system, and edges encode the interaction between individual objects. Previous studies have implemented a wide range of measures for studying networks \cite{rider_2005}, including node centrality \cite{ficara2021correlation}, clustering coefficients \cite{Masuda_2018} and path lengths between nodes \cite{melnik2016simple}. In this paper, we focus on analyzing node centrality in directed weighted networks.

Node centrality measures the influence of a node in the entire network by assigning rankings and numbers to nodes within the network based on their network position. Centrality measures enable us to detect various real-world phenomena, including the identification of banks that are too-connected-to-fail \cite{GOFMAN2017113} and decisions regarding human capital or education \cite{RePEc:cpr:ceprdp:10631}.  Different centrality measures have been developed to capture different behaviors.  For example, betweenness centrality has been widely used in social networks \cite{Lee_2021}, eigenvector centrality has proved to be useful in temporal networks \cite{taylor2016eigenvectorbased}, optimal percolation centrality has been widely used in large networks \cite{delima2020estimating}, and K-core centrality has been widely used in dynamic networks \cite{Liu_2020}.

In this paper, we are interested in studying the propagating properties of a network since influences of individuals in many real-world complex systems have the potential to propagate over the entire system. This phenomena can be meaningfully illustrated on the global trade network, where globalization of trade and intertwined economies around the world can cause economic perturbations originated in a single country propagate elsewhere. 
Moreover, since world economies are exhibiting increasing levels of local heterogeneity and global interdependency \cite{Serrano_2007}, it is crucial to take a topological standpoint on defining the centrality measure for understanding propagating effects in trade networks. Thus, we define a new centrality measure using persistent homology, namely the \emph{quasi-centrality} (Definition \ref{def: quasi-central nonno}) to determine the ranks of the nodes in directed weighted networks based on this property.  

Persistent homology is a central tool used in topological data analysis (TDA), which is a burgeoning field in math and data analysis where concepts from algebraic topology are used to simplify, summarize, and compare complex data sets. 
TDA has found successful applications in neuroscience \cite{DBLP:journals/ficn/Dabaghian20}, biological models \cite{Topaz_2015} machine learning \cite{DBLP:journals/corr/ChepushtanovaEH15}, and other related fields in statistics, math, physics, and biology. However, while many applications of TDA have focused on simplifying and identifying the intrinsic shapes of complex data sets, it has not been extensively applied to analyze networks. Hence the significance of our work is shown by utilizing TDA concepts in examining network properties. 

The quasi-centrality measure (Definition \ref{def: quasi-central nonno}) is defined based on the idea that if a node plays a crucial role in the connectivity of the network, we would expect perturbations originating from this node to propagate significantly through the network. Roughly, we measure a node's role in the overall connectivity of a directed network by computing the difference between the connectivity of the network before and after deleting this node. In particular, we use the ``size'' of homology groups as a proxy to determine the connectivity of a directed network (Remark \ref{def: rem}). 

 We show that the quasi-centrality is optimal for measuring the influence of perturbations originated from individual nodes by (1) computing the quasi-centralities on the star-shape network in Example \ref{ex: example} and comparing with existing centrality measures, and (2) analyzing the importance of country-industry pairs in the Asia machinery production network (Section 4) by computing the quasi-centralities and comparing with existing centrality measures.

Moreover, given a directed weighted network $G$, we introduce a method (Definition \ref{def: hierarc}) that extracts a hierarchical representation of nodes in $G$ based on their topological impacts. This method combines the bottleneck distance (Definition \ref{defn:bottleneck}), a tool used in TDA, with hierarchical clustering, a method used in cluster analysis that determines similarity between objects. 

We present our findings as follows: in Section 2, we provide background knowledge in networks and existing centrality measures, algebraic topology, and TDA (specifically persistent homology).
In Section 3, we define quasi-centrality (Definition \ref{def: quasi-central nonno}). 
In Section 4, we illustrate the practicality of  quasi-centrality by computing it on the Asia machinery production network. 
In Section 5, we present our method that determines the hierarchy of nodes in a given directed weighted network, and analyze the hierarchical dendrograms extracted from the  Asia machinery production networks. And finally in Section 6, we summarize our results and discuss future directions. 

\paragraph{Data and implementations:} Our data sets and software are available on \url{https://github.com/lindahe8989/A-topological-centrality-measure-for-directed-networks}. 

For plotting persistence diagrams and hierarchical dendrograms, we partially used \url{https://github.com/fmemoli/PersNet} for reference.

\section{Preliminary}
In this section, we provide background knowledge in algebraic topology and persistent homology, as well as related work in network analysis. We begin by introducing networks. 

\subsection{Networks} 
\begin{definition}
A \emph{network} is a pair $(X, w_X)$ where $X$ is a finite set and $w_X:X \times X \rightarrow \mathbb{R}$ is called the weight function. An \emph{undirected network} is a network such that $w_X(x_1,x_2)=w_X(x_2,x_1)$ for all $x_1, x_2 \in X$, and a \emph{directed network} is a network such that $w_X(x_1,x_2) \not = w_X(x_2,x_1)$ for some $x_1, x_2 \in X$. 
\end{definition}

Our definition of quasi-centrality (Definition \ref{def: quasi-central nonno}) is restricted to dissimilarity networks, introduced below: 

\begin{definition}
A \emph{dissimilarity network} is a network $G=(X, w_X)$ where $w_X: X \times X \rightarrow \mathbb{R}$ is called a \emph{dissimilarity function}, i.e. a map such that $w_X(x, x')=0$ if and only if $x=x'$ for all $x, x' \in X$. 
\end{definition}

\begin{remark}
A dissimilarity network can be mathematically represented as a directed graph without self-loops.
\end{remark}

Given a directed network $G=(X,m_X)$, the dissimilarity function $m_X(x_1, x_2)$ describes the intensity of interaction between the two nodes. We can associate a distance metric $w_X$ to $G$ such that if two nodes interact frequently, i.e., $m_X(x_1, x_2)>>0$, then they are close to each other in this distance metric, i.e., $w_X(x_1, x_2) \approx 0$. We adapt the definition of \emph{effective distance} \cite{brockmann_helbing_2013}\footnote{Successful applications include \cite{Iannelli_2017}, \cite{brockmann_helbing_2013}, etc.}, which has been proved to have successfully interpreted the relationship between distance and the volume of interaction between two nodes into our distance metric $m_X$ as follows: 

\begin{definition}
Let $G=(X, w_X)$ be a (dissimilarity) network. Define $\gamma(G)$ to be $(X, m_X)$ where $m_X: X \times X \rightarrow \mathbb{R}$ is given by: 

\[m(x_i, x_j)=\begin{cases}
1-\log \frac{w(x_i, x_j)}{\sum_{k \not = j}w(x_k, x_j)} \geq 1  \text{ if } i \not = j\\ 
0 \text{ if } i=j.
\end{cases} \]
\label{def: gamma weight}
\end{definition}

Hence in $\gamma(G)$, two vertices $x_i,x_j$ with large volume of interaction $w(x_i, x_j)$ will correspond to a closer distance $m(x_i, x_j)$. 

\begin{remark}
The notion of effective distance will be used heavily in our definition of quasi-centrality (Definition \ref{def: quasi-central nonno}). 
\end{remark}

\begin{remark}
For simplicity in our calculation, when $w(x_i,x_j)=0$, we choose a large number for $m(x_i,x_j)$, i.e., we let the distance between nodes $x_i, x_j$ in the new metric $m$ be very far apart so that it is almost indistinguishable when no relation exists between the two nodes. In Example \ref{ex: example}, we let $m(x_i,x_j)=24.026$. 
\end{remark}

\subsection{Existing centrality measures}
Given a network $G=(X, w_X)$, we can study the qualitative features of the network. One such interesting feature is to detect the importance of individual nodes in the network. In particular, \emph{node centrality} measures have been extensively applied on empirical networks to detect ranks and influence of individual nodes. 
Here we summarize existing measures of node centrality for directed networks. In Section 5 we will use these as a comparison with quasi-centrality (Definition \ref{def: quasi-central nonno}).

In the following definitions, we let $G=(X, w_X)$ be a directed network and $A_{ij}$ be its adjacency matrix.  

\begin{itemize}
    \item \emph{Degree centrality}: for a directed network  $G$, the \emph{in-degree centrality} and \emph{out-degree centrality} for a node $i \in X$ are defined to be the number of edges going in or out of $i$. 
    \item \emph{Katz centrality}: for a directed network  $G$, the  the Katz centrality $x_i$ for node $i \in X$ is defined to be $x_i=\alpha \sum_j A_{ij}x_j+\beta$, where $\alpha$ and $\beta$ are constants and $A_{ij}$ is an element of the adjacency matrix.
    \item \emph{Pagerank centrality}: for a directed network  $G$, the Pagerank centrality $x_i$ for node $i \in X$ is defined to be $x_i=\alpha \sum_j A_{ij}\frac{x_j}{k_j^{out}}+\beta$, where $\alpha$ and $\beta$ are constants, $A_{ij}$ is an element of the adjacency matrix, and $k_j^{out}$ is the out-degree of the node $i$. 
    \item \emph{HITS hubs and authorities centrality}: The HITS hubs and the HITS authorities centrality measures are defined such that one measure depends on the centrality measure determined by the other measure.

    Let $\alpha$ and $\beta$ be non-negative constants. For a directed network  $G$, the HITS-authority centrality $x_i$ for node $i \in X$ is defined to be the sum of the hub centralities $y_j$ which point to the node $i$: $x_i=\alpha \sum_jA_{ij}y_j$, and the HITS-hubs centrality $y_i$ for node $i \in X$ is defined to be the sum of the authority centrality $x_j$ which are pointed by the node $i$: $y_i=\beta \sum_jA_{ji}x_j$.
\end{itemize}

We note that all of the above centrality measures are based on the adjacency matrix $A_{ij}$, showing that the centrality measure for a node $i$ is completely dependent on the direct neighbors of the node, while not taking the propagating influence of the node on the entire network into account. The goal of our research is to provide a centrality measure that takes propagating influences into account, and we do so by using algebraic topology and persistent homology, which we will provide background knowledge in the following section.

\subsection{Algebraic topology and homology}
We begin by introducing simplicial complexes, which are the foundational objects of study in algebraic topology and homology.  

\begin{definition}
Given a finite set $ T $, an \emph{(abstract) simplicial complex} $ X $ is a subset of $Pow(T)$, the power set on $ T $,  such that 
\begin{itemize}
    \item singletons $\{x\} \subseteq T$ belong to $ X$, and
    \item whenever $\sigma \in Pow(T)$ belongs to $X$, then any subset $\tau \subseteq \sigma$ belongs to $ X $ as well. 
    We refer to a proper subset $\tau$ of  $\sigma$ is as a \emph{face} of $\sigma$. 
\end{itemize}  
\end{definition}

The elements of $T$ are called the \emph{vertices} of $ X $, the two-element subsets of $ T $ that belong to $ X $ are called the \emph{edges}, and in general the $k$-element subsets of $T$ are called the ($k$-1)-simplices in $ X $. 

Given a simplex $ \sigma = \{x_0, \ldots , x_n\} $, an \emph{orientation} of $ \sigma$ is a choice of an equivalence class of the orderings of its vertices. We say that two orderings are \emph{equivalent} if they differ by an even permutation. For $n > 0$, there are thus two equivalence classes of orderings.  
A choice of one of the classes will be called an orientation of $\sigma$. If $\sigma = \{v_0, ..., v_n\}$, then $\sigma$ with a choice of an orientation will be denoted by $\sigma = [v_0: ...: v_n]$.  

We now define homology for simplicial complexes. Let $X$ be a simplicial complex, and let $X_n$ denote the set of oriented $n$-simplices in $X$, then for any dimension $n \in \mathbb{Z}_{+}$, the \emph{vector space of $n$-chains} of $X$ over the field $\mathbb{K}$, denoted by $C_n(X, \mathbb{K})$, is the free $\mathbb{K}$-vector space with basis the set of oriented $n$-simplices: 
\[C_n(X,\mathbb{K})=\left \{ \sum \alpha_i\sigma_i:  \sigma_i \in X_n, \alpha_i \in \mathbb{K}\right \}. \]

\begin{definition}
Given any $n$-chain vector space $C_n(X,\mathbb{K})$, where $n \in \mathbb{Z}_{+}$, we define the boundary map $\partial_n: C_n(X,\mathbb{K}) \rightarrow C_{n-1}(X,\mathbb{K})$ as the the linear transformation specified on the generators as
\[\partial_n([v_0,...,v_n])=\sum_{i=1}^{n} (-1)^i[v_0,...,v_{i-1},v_{i+1},...,v_n].\]
\end{definition}

Informally, $\partial_n$ maps each $(n+1)$-simplex to its boundary comprised of its faces. It is easy to check that the boundary of a boundary is empty, and hence the linear maps $\partial_n$ satisfy the property that composing any two consecutive boundary maps yields the zero map: 

\begin{prop}\cite{Munkers84}
\label{prop: chain linear transform}
For all $n \in \mathbb{Z}_{+}$, $\partial_{n-1} \circ \partial_{n}=0$. 
\end{prop}

We axiomatize the properties of $C_n(X; \mathbb{K})$ in the definition of a chain complex. Namely, a \emph{chain complex of vector spaces} $A_{\bullet}$ is
a collection of vector spaces $\{ A_n\}$ such that 
the linear transformations $\partial_n: A_n \rightarrow A_{n-1}$ satisfy $\partial_{n-1} \circ \partial_{n}=0$ for all $n \in \mathbb{Z}_+$.
From the chain complex $A_{\bullet}$, one may define the following two subspaces: 
\[\emph{k-cycles}: Z_k(A_{\bullet}):= \ker(\partial_k)=\{a \in A_k: \partial_ka=0 \}, \] 
\[\emph{k-boundaries}: B_k(A_{\bullet}):= \im(\partial_k)=\{a \in A_k: a= \partial_{k+1}b\}. \]

\begin{remark}
By Proposition \ref{prop: chain linear transform}, the image of $\partial_{n+1}$ is contained in the kernel of $\partial_{n}$, so we can take the quotient of kernel($\partial_{n}$) by the image($\partial_{n+1}$).  
\end{remark}
 
\begin{definition}
For a chain complex $A_{\bullet}$, the $n$th \emph{homology group} $H_n(A_{\bullet})$ is defined as the quotient
\[H_n(A_{\bullet})=\text{ker}(\partial_n)/\text{im}(\partial_{n+1}).\]
\end{definition}
Roughly, we can think of homology groups as the subspace of homotopy classes of cycles in $Z_k(A_{\bullet})$ that are not the boundaries of elements of $B_{k+1}(A_{\bullet})$. 
The quotient vector space $H_n(A_{\bullet})$ is called the \emph{k-th homology of the chain complex} $A_{\bullet}$.

\begin{definition}
The dimension of $H_n(A_{\bullet})$ over $ \mathbb{K}$ is called the \emph{$n$-th Betti number of} $A_{\bullet}$, denoted by $\beta_n(A_{\bullet})$.
\end{definition}

\begin{example} 
Suppose $X$ is the simplicial complex homeomorphic to the boundary of a triangle with vertices $v_0, v_1$, and $v_2$. Then $C_0(X,\mathbb{K}) = \mathbb{K}^3$, generated by the vertices $[v_1], [v_2], [v_3]$, and $C_1(X, \mathbb{K}) = \mathbb{K}^3$, generated by the edges $[v_0, v_1], [v_1, v_2], [v_0, v_2]$, and  $C_n(X, \mathbb{K}) = 0$ for all $n > 1$. Then
\[H_0(X) = \frac{Z_0(X)}{B_0(K)}= \frac{\langle [v_0], [v_1], [v_2] \rangle}{\langle [v_0]-[v_1], [v_0]-[v_2] \rangle}=\mathbb{K}, \]
\[H_1(X) = \frac{Z_1(X)}{B_1(K)}= \frac{\langle [v_0, v_1] - [v_0, v_2] + [v_0, v_2] \rangle}{\langle 0 \rangle}=\mathbb{K}. \]
 
All other homology groups are trivial because $ C_n(X; \mathbb{K}) = 0 $ for $ n \geq 2$. The \emph{geometric realization}\footnote{See Section 1.8 in \cite{rabadan.raul.blumberg.20}: vaguely, the geometric realization $|X|$ of a simplicial complex $X$ is the operation that interprets each algebraic $n$-simplex as a standard topological $n$-simplex gluding together in a ``nice'' way. } $|X|$ of this simplicial complex $X$ is homotopic\footnote{See Section 1.8 in \cite{rabadan.raul.blumberg.20}: vaguely, two topological spaces $X$ and $Y$ are homotopic if one can be continuously deformed into the other.} to $S^1$, so this example computes the homology of $S^1$. 

Note that the dimension $ H_0 (S^1)$ corresponds to the number of path components of $S^1$, and the dimension of $H_1 (S^1) $ corresponds to the number of 1-dimensional cycles in $S^1$. 
In general, if we have a $n$-dimensional sphere $S^n$, then $H_n (S^n)=\mathbb{K} $ and  $H_0 (S^n)=\mathbb{K} $ and all other homology groups are trivial. 
\end{example}

In fact, the homology groups have a very nice interpretation: for all simplicial complexes $ X $, $ H_0(X) $ is the free $\mathbb{K}$-vector space on the set of \emph{0-dimensional cycles} of $X$ modulo boundaries. i.e., two vertices are equivalent if there exists a sequence of edges between them, which can be visualized as \emph{path components} in $|X|$; $H_1(X)$ is the free $\mathbb{K}$-vector space on the set of \emph{1-dimensional cycles} of $X$ modulo boundaries, which can be visualized as \emph{1-dimensional holes} in $|X|$; more generally, $H_k(X)$ is the free $\mathbb{K}$-vector space on the set of \emph{k-dimensional cycles} of $X$ modulo boundaries, and  can be visualized as \emph{k-dimensional holes} in $|X|$.

\begin{fact}
Let $X$ and $Y$ be topological spaces\footnote{In fact, we can define homology for ``sufficiently nice" topological spaces, not just abstract simplicial complexes.}, then a continuous map $ f: X \to Y $ induces a $ \mathbb{K}$-linear map $ H(f): H_*(X; \mathbb{K}) \to H_*(Y; \mathbb{K}) $. 
\end{fact}

\begin{remark}
Let $X$ and $Y$ be topological spaces. 
If $ f, g: X \to Y $ are homotopic, i.e., they can be continuously deformed into each other, then they induce the same $\mathbb{K}$-linear homomorphisms on the respective homology groups. 
\end{remark}

\begin{remark}
If two topological spaces $X,Y$ are homotopy equivalent\footnote{Two topological spaces $X$ and $Y$ are homotopy equivalent if there exists a pair of continuous maps $f: X \rightarrow Y$ and $g: Y \rightarrow X$ such that $f \circ g$ is homotopic to $\text{id}_Y$ and  $g \circ f$ is homotopic to $\text{id}_X$.}, then $X$ and $Y$ have isomorphic homology groups. 
\end{remark}

\subsection{Persistent homology} 
Computing the homology of finite simplicial complexes can be simplified to computing linear algebra. However, most empirical data sets do not exist in the form of simplicial complexes. Here we introduce the method in which we build simplicial complexes from directed networks.

Given a directed network $G$, one can induce a sequence of  simplicial complexes $\mathcal{D}_{\delta, G}$ for a sequence of values of parameter $\delta \in \mathbb{R}_+$ such that $\mathcal{D}_{\delta=m, G} \subseteq \mathcal{D}_{\delta=n, G}$ for all $m \leq n$, where $\mathcal{D}_{\delta, G}$ is called the Dowker sink complex (Definition \ref{def: dowker}), introduced first in \cite{chowdhury_mémoli_2018}: 

\begin{definition} 
Given a network $G=(X, w_X)$ and $\delta \in \mathbb{R}$, define $R_{\delta, G} \subseteq X \times X$ as:

\[R_{\delta, G} := \{(x, x'): w_X(x, x') \leq \delta \}. \]

Using $R_{\delta, G}$, we build a simplicial complex $\mathcal{D}_{\delta, G}$ called the \emph{Dowker sink complex} as:

\[\mathcal{D}_{\delta, G}:= \{ \sigma \in Pow(X): \text{ there exists } p \in X \text{ s.t. } (x_i,p) \in R_{\delta, G} \text{ for each } x_i\in  \sigma  \}. \]

The node $p$ is called the \emph{sink} for the simplex $ \sigma$, and one can check that $\mathcal{D}_{\delta, G} \subseteq \mathcal{D}_{\delta', G}$ for all $\delta \leq \delta'$.
\label{def: dowker} 
\end{definition}

\begin{remark}
Let $G=(X, w_X)$ be a dissimilarity network, then $R_{\delta=0, G}$ is precisely the set of $|X|$ vertices, i.e., 0-simplices. 
\end{remark}

We note that the Dowker sink complex and the \emph{Vietoris-Rips complex} \cite{otter_porter_tillmann_grindrod_harrington_2017}, one of the most prevalent tools used in TDA, are analogous. While the Vietoris-Rips complex induces a filtration of simplicial complexes from a point cloud data set, the Dowker sink complex induces a filtration of simplicial complexes from a directed network. 
 
\begin{definition}
A \emph{filtration of simplicial complexes} is a set of simplicial complexes $X_1, X_2, ..., X_n$ such that each simplicial complex is contained in its successive simplical complex: 
\[ X=X_1 \subseteq X_2 \subseteq  ... \subseteq  X_n.\]
$X$ is called a \emph{filtered simplicial complex}.
\end{definition}

\begin{example}
Let $G=(X, w_X)$ be a directed network in a finite metric space, and let the simplicial complexes $\mathcal{D}_{\delta_1, G}$, $\mathcal{D}_{\delta_2, G}$, ..., $\mathcal{D}_{\delta_n, G}$ be the Dowker sink complexes at $\delta_1 \leq \delta_2 \leq  ... \leq \delta_n$ respectively, then the sequence 

\[ \mathcal{D}_{\delta_1, G} \subseteq \mathcal{D}_{\delta_2, G} \subseteq \mathcal{D}_{\delta_3, G} \subseteq ... \subseteq \mathcal{D}_{\delta_n, G}=\mathcal{D}_{G}\] 
is a filtration of simplicial complexes, and $\mathcal{D}_{G}$ is a filtered simplicial complex.
\end{example}

One can compute the homology groups of the simplicial complex $\mathcal{D}_{\delta, G}$ for a given value of $\delta$. 
As the value of $\delta$ increases, one can record the birth (the value of $\delta$ at which a homological feature appears) and death (the value of $\delta$ at which a homological feature disappears) values for all homological features (path components, 1-cycles, etc.) that appear in $\mathcal{D}_{G}$. 

\begin{remark}
The definitions of birth, death values, and a persistence barcode of a homological feature in a filtered simplicial complex are presented in Section 4 of \cite{otter_porter_tillmann_grindrod_harrington_2017}. We will omit these definitions since they are not necessary for the work presented in this paper. 
\end{remark}

Thus, given a filtered simplicial complex $X$, we can associate $X$ a union of \emph{persistence barcodes}, each representing the birth and death values of a homological feature of $X$. Persistence barcodes record changes in the $n$-dimensional homology group $H_n$ for some $n \in \mathbb{Z}_{\geq 0}$ when one simplicial complex $X_i$ is included in its successor $X_{i+1}$ in the filtration. In particular, each persistence barcode $[b,d)$ represents the lifetime of a generator of a homology group $H_n$, where $\delta=b$ is the birth value of the feature and $\delta=d$ is the death value of the feature in the filtered simplicial complex $X$. The value $d-b$ is called the \emph{persistence} of the feature.
 
\begin{definition}
A \emph{persistence diagram} $\textbf{D}$ is a union of a finite set of points above the diagonal $D = \{(x, y) \in \mathbb{R}^2 | x = y\}$ in ${\mathbb{R}_{\geq 0}}^2$ and the entire diagonal $D$. 
\end{definition}

\begin{remark}
Each point $(x,y) \in \textbf{D}$ represents the birth and death $\delta$ values of a persistent barcode $[b,d)$ with $x=b$ and $y=d$. 
\end{remark}

\begin{remark}
Roughly, a persistence diagram inherits the same data as the persistence barcodes. However, the persistence diagram $\textbf{D}$ contains the diagonal $D$, which will be useful for comparing the similarity of data sets. Specifically, a matching between two persistence diagrams map each point of a persistence diagram to another point of the other persistence diagram or to the diagonal, For more concrete details, see Section 5.1. 
\end{remark}

Generally, given any dissimilarity network $G(X, w_X)$, we denote $\Dgm_{n}(G)$ as the $n$-dimensional persistence diagram for the filtered simplicial complex $\mathcal{D}_{G}$. Moreover, we denote  \[\textbf{P}_n(G):= \{[b_1, d_1),[b_2, d_2),...,[b_k, d_k) \}\]  to be the set of n-dimensional persistence barcodes for the filtered simplicial complex $\mathcal{D}_{G}$.

\section{Quasi-centrality measure}

In this section, we present our definition of \emph{quasi-centrality} for a directed dissimilarity network based on persistent homology and the Dowker sink complex, which is given in Definition \ref{def: quasi-central nonno}.

\begin{definition}
\label{def: function f gamma}
Let $G=(X, w_X)$ be a (dissimilarity) network and let $x \in X$, define ${ f(G, x)=(X\setminus\{x\}, w_X)} $ to be the sub-network induced by deleting $x$ and all edges incident to $x$ in $G$.
\end{definition}  

\begin{remark}
We note that for a trade network $G=(X,w_X)$, where $w_X(x_1,x_2)$ represents the trade volume between individual countries or industries $x_1, x_2 \in X$, the induced network ${f(G,x)}$ models an embargo trade barrier\footnote{For our purposes, an embargo on $ x $ refers to the situation in which all other countries/nodes refuse to trade with $x$.} on $x$. 
\end{remark}

\begin{definition}
Let $G=(X, w_X)$ be a (dissimilarity)  network and $\gamma(G)=(X, m_X)$ be defined as in Definition \ref{def: gamma weight}. Given, $x \in X$, define $\mu(x)$ to be the minimum distance between $x$ and any other node $x' \in X$ with respect to the metric $m_X$.  
\label{def: muu}
\end{definition}

\begin{definition}
Let $G=(X, w_X)$ be a (dissimilarity) network and let $x \in X$. We define the \emph{quasi-centrality} $C(x)$ as follows: 

\[C(x)=\sum_{b \in \textbf{P}_0(f(\gamma(G),x))} l(b) -\sum_{b \in \textbf{P}_0(\gamma(G))} l(b) + \mu(x) \]
where $b$ represents a persistence interval in $\textbf{P}_0$ and $l(b)$ is the persistence of the interval, i.e., the length of the interval, and $\gamma(G)$ is given in Definition \ref{def: gamma weight} and $\mu(x)$ is given in Definition  \ref{def: muu}. 
\label{def: quasi-central nonno}
\end{definition}

\begin{remark}

We first give an intuitive explanation of what $C(x)$ measures before proving that $C(x)$ is nonnegative for all $x \in X$. We note that $\gamma(G)$ not only can be treated as a topological space where $\mathcal{D}_{\gamma(G)}$ encompasses the data of its homological features, but it also can be treated as a metric space since it comprises the data of the pairwise distances between nodes in the network. 

Roughly speaking, we think of the sum $\sum_{b \in \textbf{P}_0(\gamma(G))} l(b)$ as a measure of how disconnected $\gamma(G)$ is, i.e., if $\sum$ is large, then $\gamma(G)$ is considered to be highly disconnected and vice versa. In particular, since the dimension of the 0-th homology group of a simplicial complex $X$ measures the number of path components (disconnected components) in $|X|$, we consider $\sum_{b \in \textbf{P}_0(\gamma(G))} l(b)$ to be the ``size" of the 0-th homology group of $\gamma(G)$,
and $\gamma(G)$ is considered to be highly disconnected if the sum $\sum$ is large and vice versa. 
 Similarly, the sum $\sum_{b \in \textbf{P}_0(f(\gamma(G),x))} l(b)$ as a measure of how disconnected $f(\gamma(G),x)$ is. We think of the difference as measuring of how much node $x$ contributes to the overall connectivity of the network. 

Intuitively, a network should become `more disconnected' after deleting a node, which is made precise by the following theorem. 
\label{def: rem}
\end{remark}
 
\begin{thm}
For a (dissimilarity) network $G=(X, w_X)$, $C(x)$ is nonnegative for all $x \in X$.
\end{thm}

\begin{proof}
Recall that each $[0,d_j) \in \textbf{P}_0(\gamma(G))$ represents the lifetime of a path component in $\gamma(G)$, hence by definition one connected component $\alpha$ joins another connected component $\beta$ at $\delta=d_j$. 
Therefore, at least two $0-$simplices $[x_1] \in \alpha$ and $[x_2] \in \beta$ belong to a 1-simplex $\pm [x_1:x_2]$ at $\delta=d_j$. 
By definition of the Dowker sink complex, there exists $p \in X$ such that $m(x_1, p) \leq d_j$ and $m(x_2, p) \leq d_j$.

If $x \in \{x_1,x_2,p\} $, then the 0-th persistence barcode in $\textbf{P}_0(f(\gamma(G), x)))$ representing the lifetime of path component $\alpha$ or $\beta$ will have a death resolution $\delta$ greater than or equal to $d_j$. 
If $ {x \not \in \{x_1,x_2,p\}} $, then the 0-th persistence barcode in $\textbf{P}_0(f(\gamma(G), x)))$ representing the lifetime of path component $\alpha$ or $\beta$ is still equal to $d_j$.

Hence all persistence barcodes representing the lifetime of connected components have death values in $\textbf{P}_0(\gamma(G))$ greater than or equal to the death values of corresponding components in $\textbf{P}_0(f(\gamma(G), x)))$. Therefore 
\[\sum_{b \in \textbf{P}_0(f(\gamma(G),x))} l(b)+\mu(x) \geq \sum_{b \in \textbf{P}_0(\gamma(G))} l(b),\]
and our proof is complete. 
\end{proof}

\begin{remark}
As the value of $\delta$ increases, $C(x)$ measures how often node $x$ bridges between two disconnected components. 
We think of the node $x$ as serving as a $\delta$-sink not only for a pair of nodes in a neighborhood of $x$, but rather the pair of connected components the nodes belong to.
In this way, the quasi-centrality $C(x)$ takes into account not only the local neighborhood of node $x$, but also propagating effects in said neighborhood.
\end{remark}

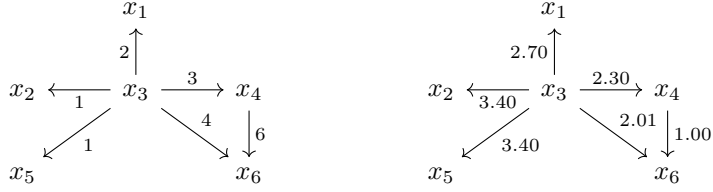
\begin{figure}
    \centering
    \hspace{-1cm}
    \begin{minipage}[b]{0.3\textwidth}
        \centering
         \begin{tikzcd}
    & x_1 \\
    x_2  & x_3 \arrow[u, "2"] \arrow[ld, "1"]  \arrow[rd, "4"] \arrow[l, "1"] \arrow[r, "3"] & x_4 \arrow[d, "6"] \\
    x_5  & & x_6
\end{tikzcd}
    \end{minipage}
    \hspace{1cm}
    \begin{minipage}[b]{0.3\textwidth}
        \centering
        \begin{tikzcd}
    & x_1 \\
    x_2  & x_3 \arrow[u, "2.70"] \arrow[ld, "3.40"]  \arrow[rd,   "2.01"] \arrow[l, "3.40"] \arrow[r, "2.30"] & x_4 \arrow[d, "1.00"] \\
    x_5  & & x_6
\end{tikzcd}
    \end{minipage}
    \hfill
    \caption{Left: $G=(X,w_X)$. Right: $\gamma(G)=(X,m_X)$. }
    \label{fig: example 1}
\end{figure}


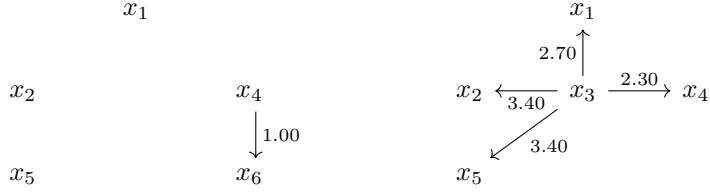
\begin{figure}
    \centering
    \hspace{-1cm}
    \begin{minipage}[b]{0.3\textwidth}
        \centering
        \begin{tikzcd} 
    & x_1 \\
    x_2  &  &    x_4 \arrow[d, shift left, "1.00"] \\
    x_5  &  & x_6
\end{tikzcd}
    \end{minipage}
    \hspace{1cm}
    \begin{minipage}[b]{0.3\textwidth}
        \centering
        \begin{tikzcd}
    & x_1 \\
    x_2  & x_3 \arrow[u, "2.70"]   \arrow[ld,  "3.40"] \arrow[l, "3.40"] \arrow[r, "2.30"] & x_4   \\
    x_5  &  
\end{tikzcd}
    \end{minipage}
    \hfill
    \caption{Left: $f(\gamma(G), x_3))$. Right: $f(\gamma(G), x_6))$. }
    \label{fig: example 2}
\end{figure}

\begin{example}
\label{ex: example}
In this example, we compute the quasi-centrality measure on nodes in an example directed weighted network $G$ as shown in Figure \ref{fig: example 1} (left), and compare the result with existing centrality measures.

The illustrations of $G$ and $\gamma(G)$ are shown in Figure 1 (left and right respectively). In the illustration of $\gamma(G)$, we leave out the edge $e_{ij}$ when $w(x_i, x_j)=0$. However, in our calculation of $C(x)$, we let the $m(x_i,x_j)$ to be a very big number when  $w(x_i, x_j)=0$ so that nodes $x_i,x_j$ in $m_X$  are very far apart that it is almost indistinguishable when no relation exists between the two nodes. In this example, we let $m(x_i,x_j) = 24.026$ to make computations more convenient.

Figure \ref{fig: example 2} illustrates the networks for $f(\gamma(G), x_3))$ and $f(\gamma(G), x_6))$, and Figure \ref{fig: example 4} illustrates the Dowker persistence barcodes for $f(\gamma(G), x_3)))$, $f(\gamma(G), x_6)))$, and $\gamma(G)$.  
 
\begin{figure*}[ht!]
\subfloat[ $\Dgm_0(\gamma(G)))$ \label{fig:PKR}]{%
      \includegraphics[ width=0.31\textwidth]{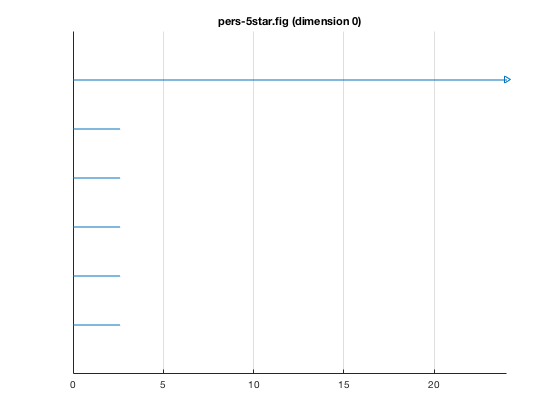}}
\hspace{\fill}
\subfloat[ $\Dgm_0(f(\gamma(G), x_3)))$ \label{fig:PKR}]{%
      \includegraphics[ width=0.31\textwidth]{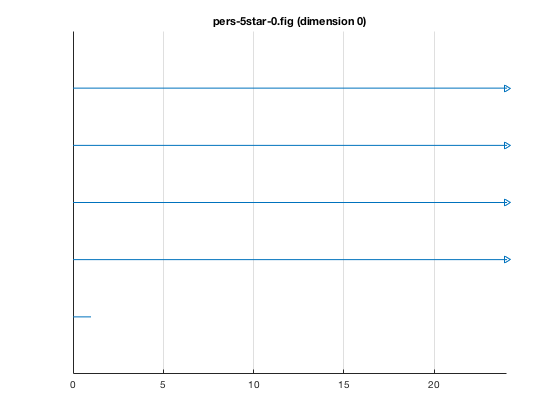}}
\hspace{\fill}
   \subfloat[$\Dgm_0(f(\gamma(G), x_6)))$ \label{fig:tie5}]{%
      \includegraphics[ width=0.31\textwidth]{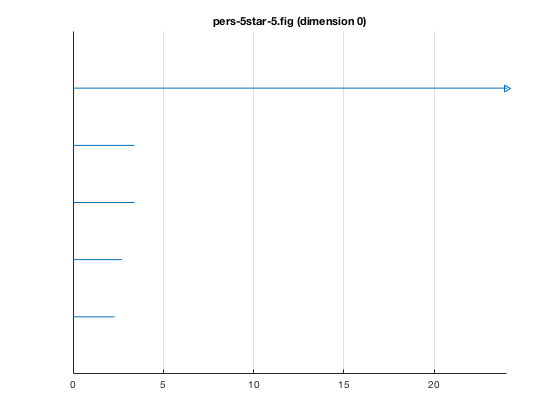}}\\
\caption{}
    \label{fig: example 4}
\end{figure*}

We illustrate the computation of $C(x_3)$ as an example. Since $\mu(x_3)=2.01$, 
\begin{align*}
    C(x_3) &= \sum_{j=1}^{n-1} d_{j}' -\sum_{j=1}^{n} d_j +2.01 \\
    &=(1+24.026\times 4)-(1+2.01+2.70+3.40+3.40+24.026)+2.01 \\
    &=65.978.
\end{align*}
Similarly, we have $C(x_1)=C(x_2)=C(x_4)=C(x_5)=0$ and $C(x_6)=0.29$. This shows that in $G$, node $x_3$ has the most propagating influence in the entire network, followed by node $x_6$. 

\begin{figure}
    \centering
\begin{tabular}{ | m{0.7cm} |  m{1.8cm} | m{1.8cm} | m{1.8cm} | m{1.8cm} | m{1.8cm} | m{1.8cm} |} 
  \hline
  Node & Quasi-centrality & Katz & Pagerank 1  & Pagerank 2 &  HITS-hubs &  HITS-authorities \\ 
  \hline
  $x_1$ & 0  & 0.38  &  0.15 &  0.50 & 0.00  & 0.11 \\
  \hline
   $x_2$ & 0 & 0.38 &  0.13 & 0.09 & 0.00  & 0.06 \\
   \hline
   $x_3$ & 65.978 &  0.23 &  0.12 & 0.09 & 0.47 & 0.00\\ 
   \hline
  $x_4$ & 0 & 0.38 &  0.16 & 0.09 & 0.53 & 0.17 \\
   \hline
   $x_5$ & 0  & 0.38  &  0.13  & 0.14 & 0.00  & 0.06\\ 
   \hline
  $x_6$ & 0.29  &  0.61 &  0.31 & 0.09 & 0.00 & 0.60\\
   \hline
  \end{tabular}
  \caption{Centrality values for $G$ obtained from the quasi-centrality, Katz centrality, HITS-hubs and authorities centralities, and Pagerank centrality.}
  \label{fig: table of quasi-cen ex}
\end{figure}

Figure \ref{fig: table of quasi-cen ex} illustrates the values obtained by other existing centrality measures. We have included two measures for the Pagerank centrality: Pagerank 1 outputs the centrality values for nodes in the original network $G$ and Pagerank 2 outputs the centrality measures for nodes in the network obtained by switching the direction of edges from $G$. 

We included Pagerank 2 for the following reason: since all existing centrality measures are based on the term $\sum A_{ij}x_j$, i.e., the centrality value $x_i$ for node $i$ only sums up the contribution values $x_j$ of nodes $j$ that points toward node $i$, by computing Pagerank 2 and comparing to the results obtained by the quasi-centrality, we essentially eliminates the possibility that the differences in the centrality values is solely due to the fact of the direction of the edges. 

We note that $x_3$ has low centrality values for all existing centrality measures, including both Pagerank centrality measures. On the contrary, we note that $x_3$ has the high quasi-centrality value, which is what we would expect since node $x_3$ holds the center position in the directed dissimilarity network $G$. Therefore in this case, the quasi-centrality better measures the propagating influences of nodes in this example. 

We expect the reason for all existing centrality measures having resulted in a low centrality value for $x_3$ for two reasons: firstly, all existing centrality measures depend on the adjacency matrix $A_{ij}$ and contain the term $\sum A_{ij}x_j$, which sums up the contribution values of nodes $j$ that links with node $i$. However, it is possible that the nodes $j$ are not central in the network and hence does not contribute to the centrality value of node $i$, which can be shown in this example. However, it is still possible that node $i$ holds the center position in the network even when its neighbors are not important in the network (such as in this example), in which all existing centrality measure fail to accurately measure its importance. Secondly, the adjacency matrix term $\sum A_{ij}x_j$ only sums up the contribution values of nodes $j$ that has an ingoing edge to node $i$ while ignoring the nodes that have an outgoing edge from node $i$, meaning that the centrality values are only taking into account of propagating influences on one side while ignoring the other.  

Generally, when measuring the propagating influence of a node, both incoming edges and  outgoing edges should positively impact the centrality value. For example, when measuring the influence of perturbations originated from a country or industry in the empirical trade production network, if a particular country or industry is both an important exporter and importer, it will have a greater influence on the entire economy compared to a country or industry that is only an important importer. The quasi-centrality in this case would give a more thorough measure of a node's propagating influence as it takes into account of the nodes' impact in both directions, whereas other centrality measures might only focus on ingoing edges or outgoing edges.

\end{example}

\section{Quasi-centrality applications}
In this section, we apply the quasi-centrality measure to an empirical network describing trade among countries and industries--the Asia machinery production network in 2007 and 2011. 
We show that the quasi-centrality measure gives a more refined evaluation of the propagating influences of nodes in the trade network than existing centrality measures.

The trade network exhibits many properties common to complex networks, including high clustering coefficient and degree-degree correlation between different vertices \cite{PhysRevE.68.015101}. 
Moreover, research suggests that random microeconomic shocks have the ability to propagate through a network with high inter-connectivity and cause “cascade effects” in the entire network and economy \cite{econometrica_2012}--a prime example is the Asiatic crisis. 
Therefore, it is important to take into account of the propagating effects of the trade network when evaluating node centrality. We later show that the quasi-centrality accurately predicts the influences of nodes in the trade network as it takes into account of such propagating effects.

We examine the machinery production network in Asia during 2007 and 2011 (network illustrations are shown in Figure \ref{fig: 2007 and 2011 Asia Network Visualization}), before and after the 2008-2009 financial crisis. Our network consists of both different intermediate products and different countries as nodes, as either counterpart plays a crucial role in supply and production chains\footnote{\label{footnote:supply}A supply/production chain is a system of intermediate transactions, suppliers, distributors, consisting of different countries and intermediate products.}, which are central for our understanding of the entire network.

\begin{figure*}[ht!]
\subfloat[ 2007 Asia machinery production network \label{fig:PKR}]{%
      \includegraphics[ width=0.5\textwidth]{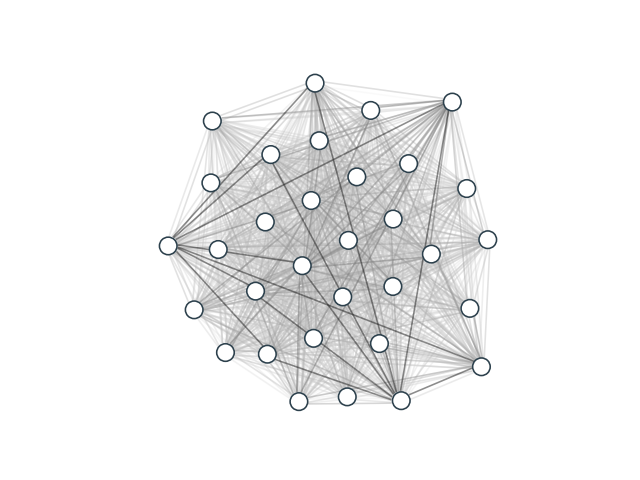}}
\hspace{\fill}
   \subfloat[2011 Asia machinery production network \label{fig:tie5}]{%
      \includegraphics[ width=0.5\textwidth]{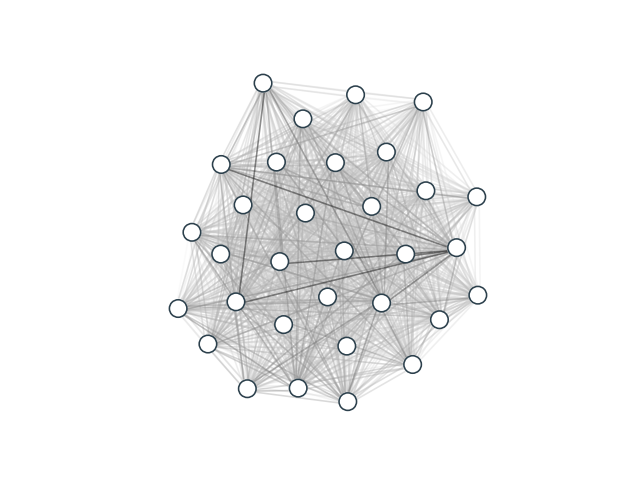}} \\
\caption{}
     \label{fig: 2007 and 2011 Asia Network Visualization}
\end{figure*}

Our data consists of $8 \times 4$ pairs of countries and industries, obtained from the OECD Inter-Country Input-Output (ICIO) Tables \cite{OECD2021}, which provides transport values between country-industry pairs. The eight countries and four industries are labeled in the following table:

\begin{center}  
\begin{tabular}{ |c|c|c|c| } 
\hline
Country Label & Country Name \\
\hline
KOR & North Korea \\
\hline
JPN & Japan \\
\hline
MYS & Malaysia \\
\hline
SGP & Singapore \\
\hline
TWN & Chinese Taipei \\
\hline
CN & China \\
\hline
VNM & Vietnam \\
\hline
THA & Thailand \\
\hline
\end{tabular} 
\end{center}

\begin{center}

\begin{tabular}{ |c|c|c|c| } 
\hline
Industry Label & Industry Name \\
\hline
C29 & Machinery and equipment \\
\hline
C30 & Computer, electronic and optical equipment \\
\hline
C31 & Electrical machinery and apparatus \\
\hline
C34 & Motor vehicles, trailers and semi-trailers \\
\hline
\end{tabular} 

\end{center}  

For brevity, we abbreviate the 32 nodes by [label of industry]-[label of the country]. 
For example, the machinery and equipment industry in Korea is referred to as KOR-C29.

\begin{remark}
In our data obtained from the OECD Inter-Country Input-Output (ICIO) Tables, imports are valued at basic prices of the country of origin, i.e., the domestic and international distribution included in goods imports in c.i.f. purchasers' prices are re-allocated to trade, transport and insurance sectors of foreign and domestic industries. 
\end{remark}

\begin{remark}
We exclude taxes paid and subsidies received in foreign countries in our network. 
\end{remark}
\begin{remark}
For the rest of Section 4, we will use ``node'' and ``country-industry pair'' interchangeably.
\end{remark}

\subsection{2007 Asia machinery production network}

We first examine nodes' propagating influences in the 2007 Asia machinery production network. We compute the quasi-centralities (Figure \ref{fig: 2007-asia-quasi-centrality}) for nodes in the network and compare with existing centrality measures (Figure \ref{fig: 2007-asia-other-centrality}) and later demonstrate that quasi-centrality successfully captures propagating effects and more accurately evaluates nodes' propagating influences in the network.

\begin{figure}
    \centering
    \includegraphics[width=13cm]{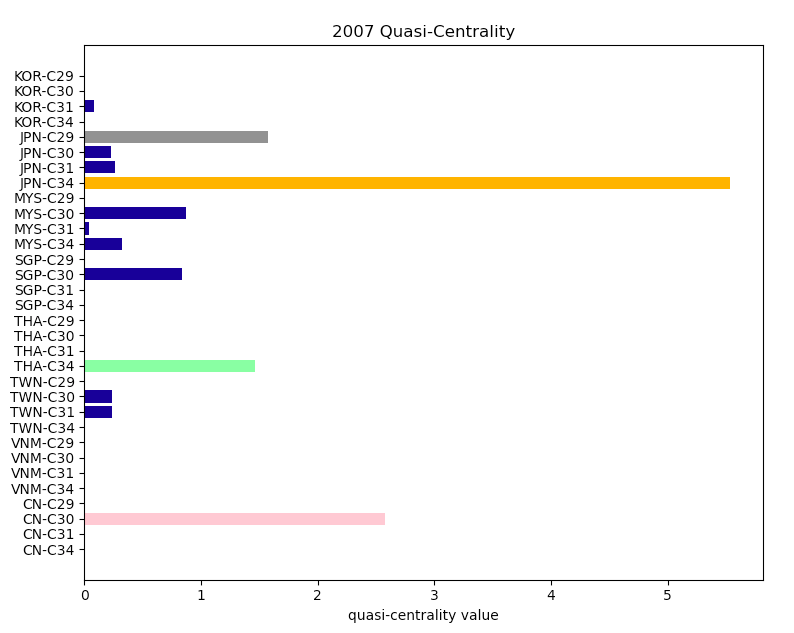}
    \caption{Quasi-centralities for the 2007 Asia machinery production network.}
    \label{fig: 2007-asia-quasi-centrality}
\end{figure}

\begin{figure}
    \centering
    \hspace{-2cm}
    \begin{minipage}[b]{0.3\textwidth}
        \centering
        \includegraphics[width=7.3cm]{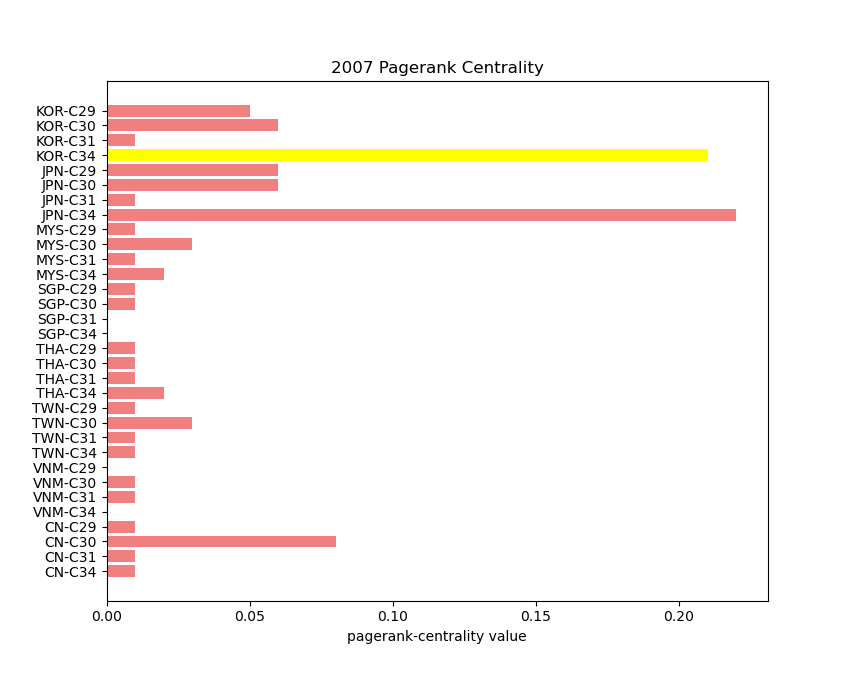}
    \end{minipage}
    \hspace{3cm}
    \begin{minipage}[b]{0.3\textwidth}
        \centering
        \includegraphics[width=7.3cm]{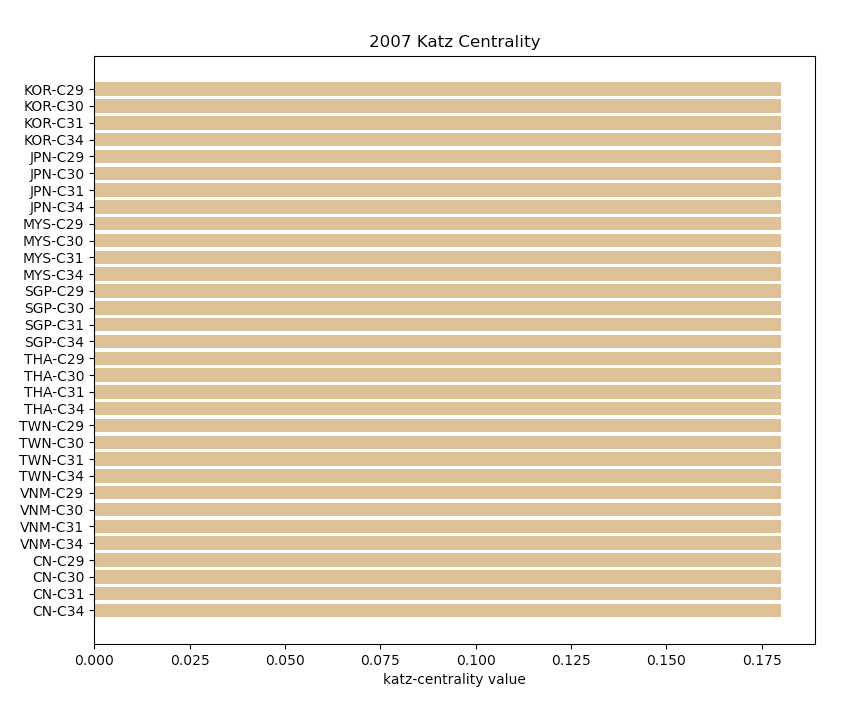}
    \end{minipage}
    \vfill
    \hspace{-2cm}
    \begin{minipage}[b]{0.3\textwidth}
        \centering
        \includegraphics[width=7.3cm]{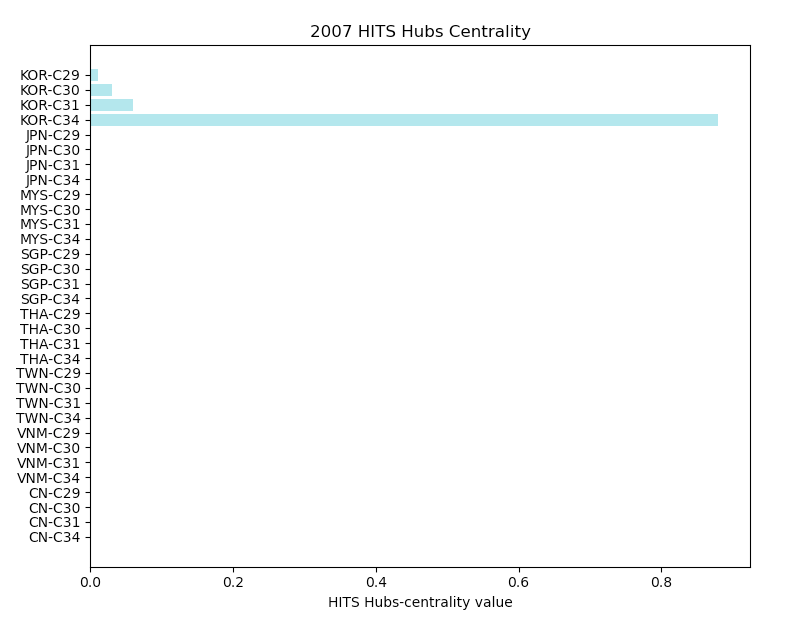}
    \end{minipage}
    \hspace{3cm}
    \begin{minipage}[b]{0.3\textwidth}
        \centering
        \includegraphics[width=7.3cm]{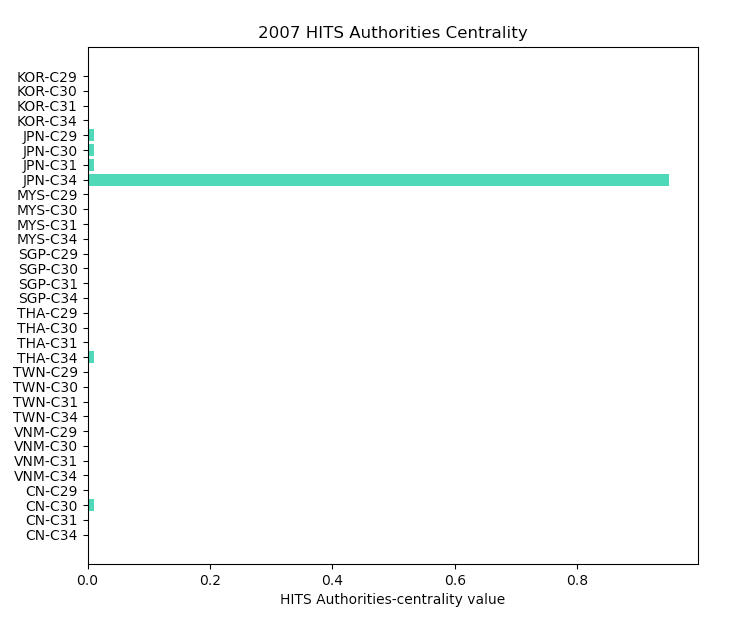}
    \end{minipage}
    \hfill
    \caption{Existing centralities (Katz, Pagerank, HITS-hubs, HITS-authority) for the 2007 Asia machinery production network.}
    \label{fig: 2007-asia-other-centrality}
\end{figure}

According to Figure \ref{fig: 2007-asia-quasi-centrality}, JPN-C34 (marked in orange) has the highest quasi-centrality value, implying that it has the highest propagating influence over the network, followed by CN-C30 (marked in pink), JPN-C29 (marked in grey), THA-C34 (marked in green), etc.
These values reflect that China and Japan were the two largest contributors to the machinery production industry in Asia in 2007 \cite{doi:10.1080/10168737.2016.1148398}.

We note that THA-C34 also has a relatively high quasi-centrality value. Although Thailand only accounted for 1.12 $\%$ of global exports of vehicle parts in 2007, while Japan, Korea, and China accounted for 10.2 $\%$, 4.05 $\%$, and 3.98 $\%$ (\cite{OEC.2021}) respectively, Thailand holds an important network position locally in Asia due to its high  inter-correlation with China, the leading player in the network. 
 According to \cite{/content/publication/9789264237087-en}, between 2000 to 2010, Thailand's exports in motor vehicles parts, electronics and electrical appliances increased by $>40$ percent, and growth in trade with China (both imports and exports) also grew substantially. 
 By 2010, China became Thailand's lead export destination.
 
 Since quasi-centrality is designed to capture how much a node contributes to the overall connectivity of the network, meaning that higher intensity of links between nodes will lead to a higher connectivity, we may also  conclude the following: if a node $a$ is highly influential in the network, and if another node $b$ exhibits strong linkage with node $a$, then node $b$ also plays a relatively important role in the network because perturbations originated from $b$ will have a high influence on $a$, and will then be highly influential to the entire network. Since Thailand is highly interconnected with China, one of the main players in the network, we may infer that perturbations originating from Thailand have a high potential to propagate over the entire network.

\paragraph{Comparison with existing centrality measures.}

We observe that the relative importance of THA-C34 according to the quasi-centrality is not reflected in other existing centrality measures. We hypothesize that this is due the fact that the existing centrality measures are dependent on the adjacency matrix $A_{ij}$, implying that the centrality value for THA-C34 is based on its local neighborhood, while the quasi-centrality evaluates the THA-C34's contribution to the connectivity of the entire network, emphasizing its interconnectedness with CN-C30.  

According to Figure \ref{fig: 2007-asia-other-centrality}, the Katz centrality returns similar centrality values for all nodes in the network, failing to capture subtleties in the influences of nodes. 
Since the trade network exhibits high clustering properties, i.e., every country-industry pair is linked with almost every other country-industry pair, the Katz centrality is unable to capture the relative differences in the influence of nodes. 
We may conclude that this is due to the fact that Katz centrality ignores the importance of the intensity of interaction between edge links, i.e. it does not take edge weights into account.

Moreover, both HITS-hubs and HITS-authorities centralities are not refined since their output values fail to evaluate the relative scale of the the influences of nodes in the network: 
they are only able detect a single node that is highly influential, while ignoring the influences of other nodes. 
This is not helpful for detecting potential souces of shocks in supply chains because we are only able to detect one single chain instead of many chains.

Of the existing centrality measures, the Pagerank centrality gives similar relative information of the influence of nodes compared to the gradation given by the quasi-centrality. 
However, we note that the Pagerank centrality outputs a high centrality value for KOR-C34 (marked in yellow), while the quasi-centrality measure for KOR-C34 is zero. 
According to \cite{doi:10.1080/10168737.2016.1148398}, even though Korea has high import value for machinery final products, its export value is relatively low--in particular, KOR-C34 has a many more incoming edges than outgoing edges. 
Since the Pagerank centrality $x_i$ for a node $i$ is $\alpha \sum_{j} A_{ij}x_j+\beta$, which only sums up the contribution value of nodes $j$ that has an edge that points toward node $i$, KOR-C34 has a high Pagerank centrality value as it has a large number of incoming edges, i.e., imports.

However, a country-industry pair's propagating influence in the entire trade network should depend on both its export and import values, as a country-industry pair with both high export and import should have a larger propagating influence than a country-industry pair with only high import values. 
Thus in this case, quasi-centrality offers a better assessment of a node's influence than Pagerank does. 
While the Pagerank centrality only remembers one direction at a time, meaning that it captures some propagating influence but not all, the quasi-centrality captures both directions and gives a more comprehensive evaluation of the propagating influence of a node.

\subsection{2011 Asia machinery production network} 
Now we proceed to analyze the propagating influences of nodes in the 2011 Asia machinery production network, following the 2007-2008 financial crises. The objective is to identify changes in power dynamics of nodes in the network compared to 2007.

\begin{figure}
    \centering
    \includegraphics[width=13cm]{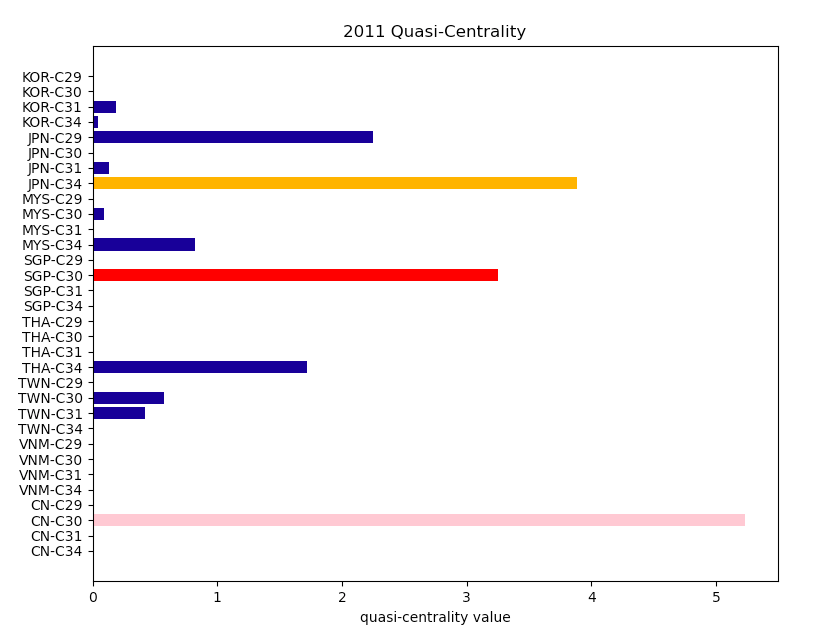}
    \caption{Quasi-centralities for the 2011 Asia machinery production network.}
    \label{fig: 2011-asia-quasi-centrality}
\end{figure}

\begin{figure}
    \centering
    \hspace{-2cm}
    \begin{minipage}[b]{0.3\textwidth}
        \centering
        \includegraphics[width=7.3cm]{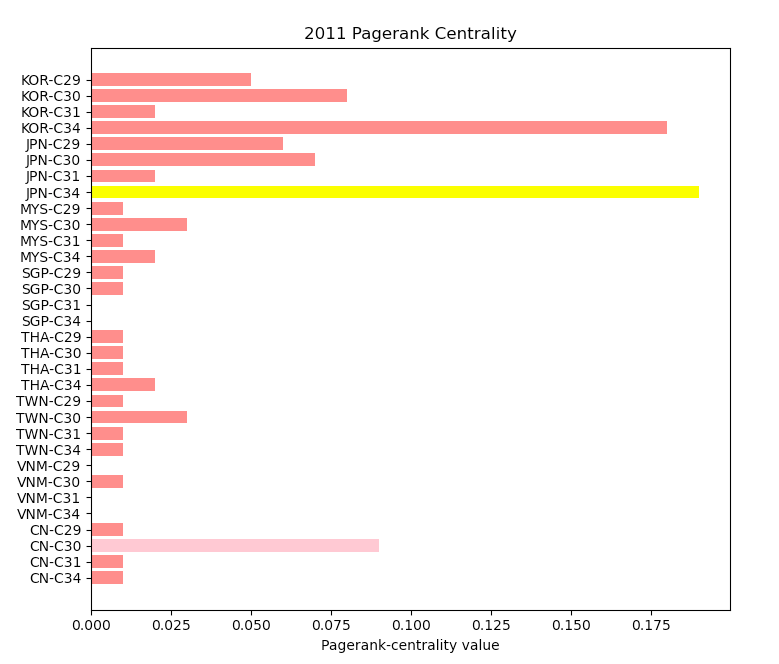}
    \end{minipage}
    \hspace{3cm}
    \begin{minipage}[b]{0.3\textwidth}
        \centering
        \includegraphics[width=7.3cm]{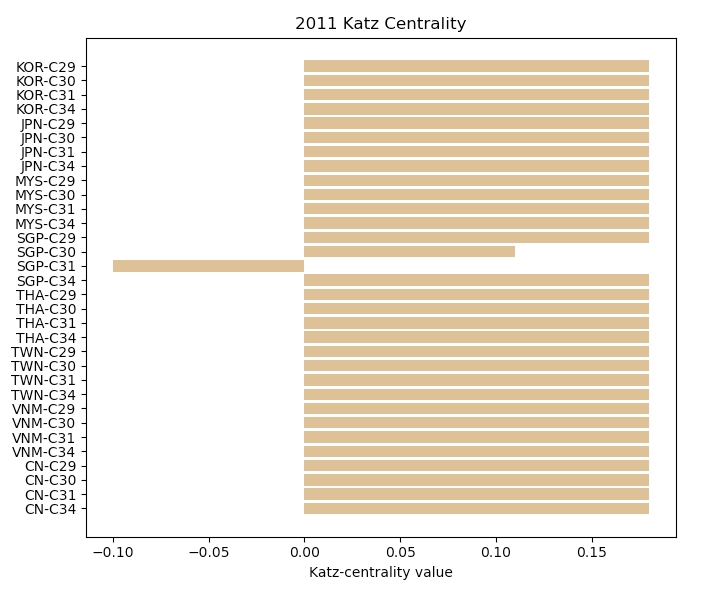}
    \end{minipage}
    \vfill
    \hspace{-2cm}
    \begin{minipage}[b]{0.3\textwidth}
        \centering
        \includegraphics[width=7.3cm]{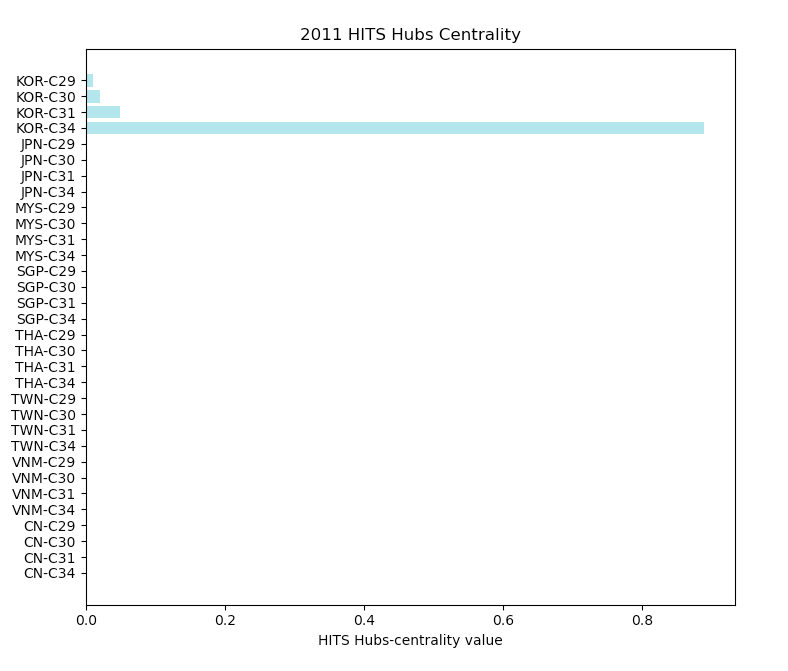}
    \end{minipage}
    \hspace{3cm}
    \begin{minipage}[b]{0.3\textwidth}
        \centering
        \includegraphics[width=7.3cm]{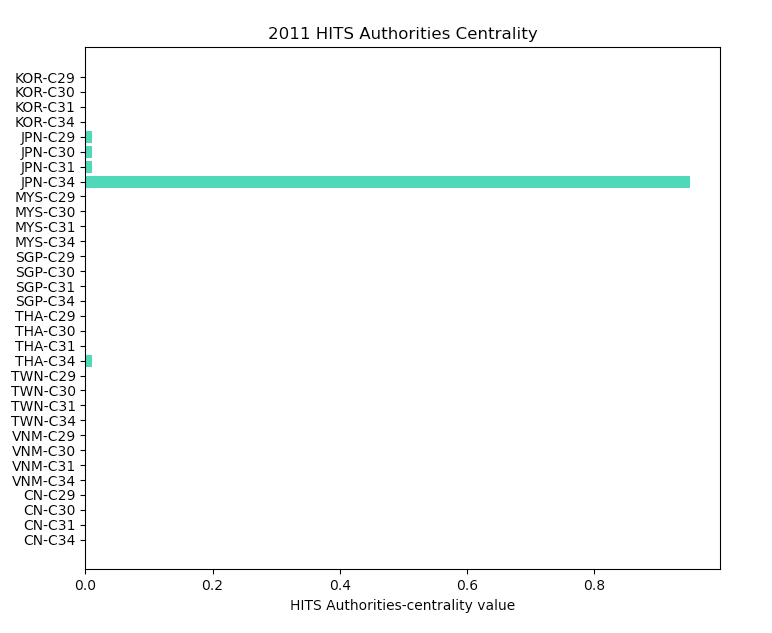}
    \end{minipage}
    \hfill
    \caption{Existing centralities (Katz, Pagerank, HITS-hubs, HITS-authority) for the 2011 Asia machinery production network.}
    \label{fig: 2011-asia-other-centrality}
\end{figure}

According to Figure \ref{fig: 2011-asia-quasi-centrality}, CN-C30 (marked in pink) has the highest quasi-centrality value, followed by JPN-C34 (marked in orange) and SGP-C30 (marked in red).  
It is worth noting the change in the quasi-centrality values for JPN-C34, CN-C30, and SGP-C30 from 2007 to 2011. While the quasi-centrality value for CN-C30 and SGP-C30 have significantly increased, the value for JPN-C34 has significantly decreased.

According to \cite{doi:10.1080/10168737.2016.1148398}, China became a dominant player in the global machinery production network in terms of \emph{both} export value and the diversity of industry-destination pairs following the 2008-2009 financial crisis. Since the quasi-centrality measures how a node contributes to the overall connectivity of the network, accounting for \emph{all} factors of import, export values, and industry-destination pairs, China's growth of impact in the network is substantially illustrated by its increase in the quasi-centrality value. 

According to \cite{doi:10.1080/10168737.2016.1148398}, Japan's machinery production growth, on the other hand, stagnated because of two main reasons after the financial crisis of 2008-2009. First, transportation links between Japan and multiple countries weakened, specifically with Korea.
 
Second, there was a noticeable decrease in the number of product–destination pairs in Japan’s exports of machinery. Both of these factors have led to Japan's stagnation in economic growth, and the stagnation is distinctively illustrated by the change of Japan's quasi-centrality value relative to others as computed in Figure \ref{fig: 2011-asia-quasi-centrality}.

In addition to the significant power dynamic changes of Japan and China in the network, we also note that the quasi-centrality value for SGP-C30 has significantly increased compared to 2007. According to \cite{hoon_2022},  the Singapore electrical machinery and apparatus industry became more \emph{central} in numerous supply chains (see footnote \ref{footnote:supply}) after the financial crisis, thereby playing an important role as an intermediary between an abundance of production, supply, distribution, and post-sales activities of goods and services.

\paragraph{Comparison with existing centrality measures.}  

We note that China's dominance in the network and Japan's stagnation of economic growth are not reflected in the existing centrality measures. Figure \ref{fig: 2011-asia-other-centrality} shows that the HITS-hubs and authorities centrality measures are not refined in that they continue to only detect the node with the highest influence in the network (KOR-C34 and JPN-C34 respectively), and that the Katz centrality continues to output the same centrality measure for almost all the nodes. 

Moreover, the Pagerank centrality continues to measure JPN-C34 (marked in yellow) as being the most influential node in the network. We hypothesize that this is due to the fact that since the Pagerank centrality only sums up the contribution values of nodes having an incoming to JPN-C34, the decrease in product–destination pairs for JPN-C34 of \emph{exports} would not affect its Pagerank centrality. 
Furthermore, CN-C30 (marked in pink) also has a relatively low Pagerank centrality value even though it is the dominant player in the network in 2011. We hypothesize that this is again due to the fact that the Pagerank centrality only sums up the contribution values of nodes having an incoming edge to CN-C30, but CN-C30 exports far more than imports. Specifically, China's raw export value of machines globally was $ \$ 894 $ billion dollars in 2011, while its raw import value of machines globally was only $ \$ 328 $ billion dollars.

Moreover, the increase of Singapore electrical machinery industry's influence in the network is substantially reflected by its quasi-centrality value, while not reflected in the existing centrality measures, again showing how the quasi-centrality is effective in measuring a node's propagating influence in the network.

\subsection{Summary}
We computed the quasi-centralities of nodes  in the 2007 and 2011 Asia machinery production network and compared the result with the existing centrality measures. Based on empirical trade statistics, we showed that quasi-centrality is optimal for detecting propagating influences of a node in the network as it takes into account both the import and export values of a country and industry pair, as well as the node's number of industry-destination pairs and its role in intermediary transactions. 

Specifically, we showed that empirical facts support the information obtained by computing the quasi-centrality measure: a set of countries and industries such that perturbations originating from this set of nodes can propagate significantly through the network.  
 
\paragraph{Potential applications}
Trade is often adversely affected by perturbations in a single country, including natural disasters such as a hurricane, earthquake or flood, and political turmoil or an armed conflict \cite{korniyenko_pinat_dew_2017}. Understanding which collections of nodes in the network might potentially cause significant propagating impacts can aid in identifying bottlenecks, thereby making the trade network more resilient. 
We propose that using the quasi-centrality measure can ultimately serve to protect world economies.

\section{Clustering nodes by topological influences}

In section 5.1, we introduce a method (Definition \ref{def: hierarc}) that combines hierarchical clustering and persistent homology to construct a hierarchy of nodes based on their impact in the overall topology of a directed weighted network. 
In sections 5.2 and 5.3, we apply our method to the 2007 and 2011 Asia machinery production network.

\subsection{Node hierarchy using TDA} 
Before introducing our method in Definition \ref{def: hierarc}, we recall how hierarchical clustering works. 

Given a set of objects and a specified weight function \cite{kulish_2002} that computes the similarity between two objects in the set, hierarchical clustering is a method for detecting community structures within this set of objects by arranging the objects into a hierarchy of groups according to the weight function. 
For a given value of parameter $t$, hierarchical clustering groups objects that are ``at most $t$ far apart from each other'' according to the weight function.

\begin{definition}
Given a finite set $X$, a \emph{partition} $ P $ of $X$ is a collection $ P=\{B_1, ..., B_k\}$ of subsets of $X$ such that: 
\begin{enumerate}[nolistsep]
    \item For all $i \not = j$,  $B_i \cap B_j = \emptyset $, and 
    \item The union $\cup_{i=1}^{k}B_i =X$.
\end{enumerate} 
Each $ B_i $ is referred to as a \emph{block} of $P$. We denote the set of all partitions of $X$ by $ \text{Part}(X) $. 
\end{definition}

Because the partition depends on t, we would like a family of t-partitions, hence a hierarchical dendrogram is defined to be a function that outputs a partition given an input value of t:

\begin{definition}
Let $X$ be a finite set. A \emph{dendrogram over $X$} is a weight function \\
$ {D_X: \mathbb{R}_{+} \rightarrow \text{Part}(X)} $ such that:
\begin{enumerate}[nolistsep]
    \item  For $t' \geq t$, every block of $D_X(t)$ is contained in a block of $D_X(t')$.
 \item  There exists $t_F\in \mathbb{R}_{+}$ such that for all $t \geq t_F$ , $D_X(t)=\{X\}$. 
 \item The weight function $D_X(0)$ consists of singletons $\{x\}$ for all $x \in X$ .
 \item  For all $t\in \mathbb{R}_{+}$, there exists $\epsilon > 0$ such that $D_X(t)=D_X(t')$ for all $t' \in [t, t+\epsilon ]$.
\end{enumerate}
\end{definition}

As the value of parameter $t$ increases, the partition $D_X(t)$ gets ``coarser'', i.e., objects that are within $t$ similar to each other belong to the same block in $D_X(t)$.  

A dendrogram can be thought of as a (nested) family of partitions of $X$ into \emph{similarity classes}, with the variable $ t $ specifying the desired fineness and granularity of the partition $ D_X(t) $. Elements that belong to the same block in $D_X(t)$ for small values of $t \in \mathbb{R}_{+}$ are more similar according to the weight function $D_X$, and vice versa. 

\begin{remark}
We will not be going into detail of the computer algorithm that we use to construct our hierarchical dendrogram since the exact computer algorithm will not be necessary for explaining our work. We will use \emph{single-linkage clustering}\footnote{\label{footnote:single}You can read more concrete details about the algorithm in Section II of \cite{8862232}.} as our clustering algorithm, which is also used in Definition \ref{def: hierarc}. 
  
\end{remark}

Now we explain how we can extract a hierarchical dendrogram from a dissimilarity directed network $G=(X,w_X)$ based on nodes' impact in the topology of the network (Definition \ref{def: hierarc}).
We first introduce the set of objects included in our hierarchy:

\begin{definition}
Given a dissimarity network $G=(X,w_X)$, let $\gamma(G)$ be as defined in Definition \ref{def: gamma weight} and assume that $|X|=n$, the number of nodes in $G$. For each node $x \in X$, let $f(\gamma(G),x)$ be as defined in Definition \ref{def: function f gamma}. Define the set of objects denoted by $S_G$ in the hierarchical dendrogram associated to $G$ as:

\[S_G :=\left\{\Dgm(f(\gamma(G),x)) \mid x \in \gamma(G) \right\} \cup \{ \Dgm(\gamma(G))) \}.\]
\end{definition}

We are interested in using dendrograms to express the hierarchy of nodes in a given directed network based on nodes' topological influences in the network. 
The main idea is that the bottleneck distance between the Dowker persistence diagram associated to $G$ and the diagram associated to the induced network $G_x=(X\setminus\{x\},w_X)$ is small if $ x $ does not play an important role in the topology of $G$, and vice versa. We first introduce the set of objects $S_G$ in the hierarchical dendrogram associated to $G$:

Recall that given a directed weighted network, the Dowker persistence diagram summarizes the homological features in the Dowker complex associated to the network.
We may then conclude that given two directed weighted networks, if the two  Dowker persistence diagrams associated to the networks are similar, then the networks have similar topology.  

We measure the similarity between two persistence diagrams using the bottleneck distance (Definition \ref{defn:bottleneck}), which describes the cost of the optimal \emph{matching}\footnote{A matching $\eta$ between persistence diagrams $D_1$ and $D_2$ pairs each point in $D_1$ with a point in $D_2$ or a point on the diagonal line, and pairs each point in $D_2$ with a point in $D_1$ or a point on the diagonal.} between points of the two diagrams \cite{agami2020comparison}, introduced below:

\begin{definition}\label{defn:bottleneck}
Given two persistence diagrams $D_1$ and $D_2$, the \emph{bottleneck distance}  of $D_1$ and $D_2$ denoted by $d_{B_\infty} (D_1,D_2)$ is defined as
\[d_{B_\infty} (D_1,D_2)=\inf_{\eta: D_1 \rightarrow D_2} \sup_{x\in D_1} ||x-\eta(x) ||_\infty \]
where $\eta$ ranges over all embeddings $\eta:  D_1 \hookrightarrow D_2 \cup \Delta $, where $ \Delta $ denotes the diagonal and $||(x,y)||_\infty$ is the usual $L_\infty$ norm.  
\end{definition}

 If the bottleneck distance between two Dowker persistence diagrams is close to zero, i.e., if $d_{B_\infty} (\Dgm_1,\Dgm_2) \approx 0$, then we may conclude that two corresponding directed weighted networks have similar topological structure, and vice versa.

Using $S_G$, we may obtain the hierarchical dendrogram associated to the dissimilarity directed network $G$ as follows: 

\begin{definition}

Given a dissimilarity network $G=(X,w_X)$ with the bottleneck distance, the \emph{hierarchical dendrogram $\mathcal{H}_G$ associated to $G$} is the function
 
\begin{align*}
    \mathcal{H}_G:& \mathbb{R}_+ \rightarrow \text{Part}(S_G)\\
    &t \mapsto \text{partition of } S_G \hspace{1mm} (S_G, \text{bottleneck distance}). 
\end{align*}

And we implement the single-linkage clustering (see footnote \ref{footnote:single}) algorithm to extract the hierarchical dendrogram from the function $\mathcal{H}_G$.

\label{def: hierarc}
\end{definition}

For a given object $ x \in X $ and parameter value $t=t'$, we will denote the block containing $ x $ in $\mathcal{H}_G(t')$ by $\mathcal{B}_{\mathcal{H}_G}(x,t')$.

\subsection{Hierarchical clustering of the 2007 Asia machinery production network}
In this section, we examine the 2007 Asia machinery production network by analyzing the hierarchical dendrogram associated to the network obtained by the method we introduced in the previous section (Definition \ref{def: hierarc}). 

Figure \ref{fig: 2007 asia gamma dendrogram} shows the hierarchical dendrogram associated to the 2007 Asia machinery production network, where the $y$-axis shows the 33 \emph{nodes}\footnote{For brevity, we will refer to the objects in the hierarchy to be the nodes in the dendrogram.} in the dendrogram, corresponding to the elements in $S_G$, where each of $\Dgm(f(\gamma(G),x))$ is labeled by the node $x$ deleted, and  $\Dgm(\gamma(G))$ is labeled by STANDARD, and the values along the $x$-axis corresponds to the values of the parameter $t$ in the weight function $\mathcal{H}_G$. 
The vertical line at $t=t'$ corresponds to the blocks of the partition  $\mathcal{H}_G(t')$.

We consider the smallest value of $t$ for which $ \Dgm(f(\gamma(G), x)) $ belongs to the cluster containing STANDARD to be its topological impact on the network. 
In other words, we expect nodes in the dendrogram belonging to the block containing the STANDARD node at smaller values of $t$ in the partition $\mathcal{H}_G(t)$ to have a smaller impact in the topology of the network, i.e., if given two nodes $a,b$ and $t < t'$, and that $a \in \mathcal{B}_{\mathcal{H}_G}(\text{STANDARD},t)$ and $b \in \mathcal{B}_{\mathcal{H}_G}(\text{STANDARD},t')$ then $b$ has a greater topological influence in the network topology. 

According to Figure \ref{fig: 2007 asia gamma dendrogram}, we see that $\text{THA-C30},  \text{VNM-C34} \in \mathcal{B}_{\mathcal{H}_G}(\text{STANDARD},t=0)$, implying that the country-industry pairs THA-C30 and VNM-C34 are insignificant in the topological structure of the network. We consider these nodes as ``peripheral,'' i.e. deleting them doesn't really change the network topology. 

\begin{figure}
    \centering
    \includegraphics[width=14cm]{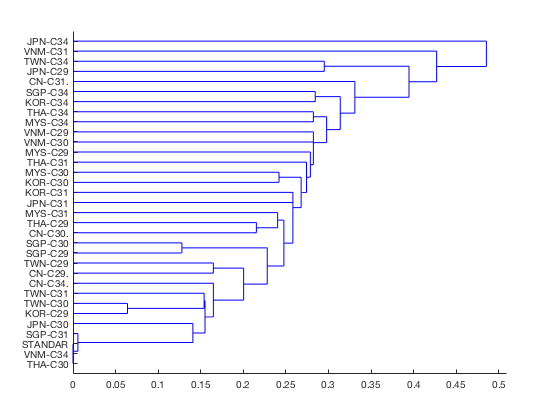}
    \caption{Hierarchichal dendrogram associated to the 2007 Asia machinery production network.}
    \label{fig: 2007 asia gamma dendrogram}
\end{figure}

On the contrary, when $t < 0.48$, $\text{JPN-C34} \not \in \mathcal{B}_{\mathcal{H}_G}(\text{STANDARD},t)$ while \\
$ {\mathcal{B}_{\mathcal{H}_G}(\text{STANDARD},0.42) = \{S_G \setminus  \{\text{JPN-C34}\} \} } $, thus we conclude that deleting JPN-C34 has the most significant impact in the topological structure of the network. 
Indeed, we see that the 1-dimensional barcodes of $\Dgm(f(\gamma(G),\text{JPN-C34}))$ (Figure \ref{fig: 2007 dendrogram 3}(b)) differs significantly from the 1-dimensional barcodes of $\Dgm(\gamma(G))$ (Figure \ref{fig: 2007 dendrogram 3}(a)): while all 1-dimensional barcodes in  $\Dgm(\gamma(G))$ have death parameters $\delta \leq 4.38$, there exists a 1-dimensional barcode in $\Dgm(f(\gamma(G),\text{JPN-C34}))$ with death parameter $\delta = 5.64 $, indicating that JPN-C34 has a significant impact in the one-dimensional homological features of $\gamma(G)$, reflecting its importance in the network topology.

\begin{figure*}[ht!]
\subfloat[ $\Dgm(\gamma(G))$. \label{fig:PKR}]{%
      \includegraphics[ width=0.5\textwidth]{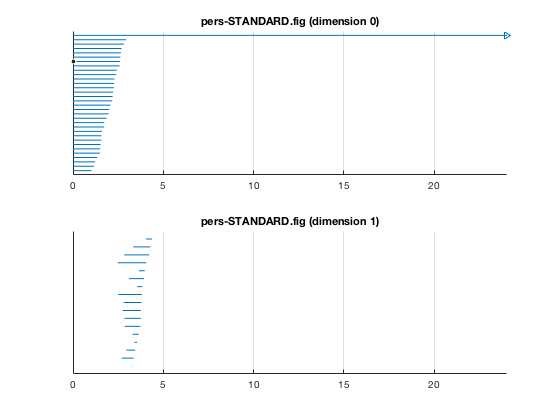}}
\hspace{\fill}
\subfloat[ $\Dgm(f(\gamma(G),\text{JPN-C34}))$. \label{fig:PKR}]{%
      \includegraphics[ width=0.5\textwidth]{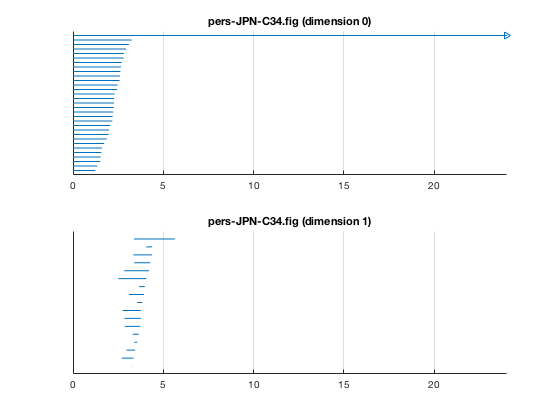}}
\hspace{\fill}
   \subfloat[$\Dgm(f(\gamma(G),\text{VNM-C31}))$. \label{fig:tie5}]{%
      \includegraphics[ width=0.5\textwidth]{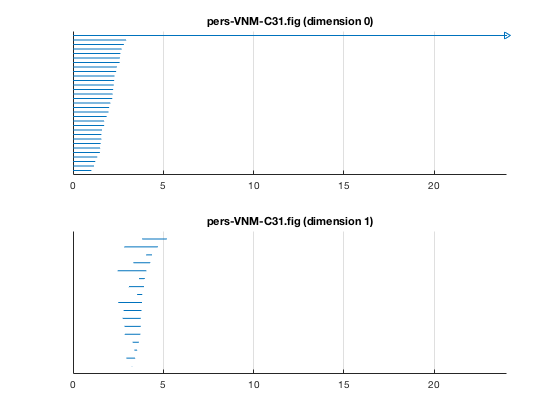}}
\hspace{\fill}
   \subfloat[$\Dgm(f(\gamma(G),\text{CN-C31}))$. \label{fig:tie5}]{%
      \includegraphics[ width=0.5\textwidth]{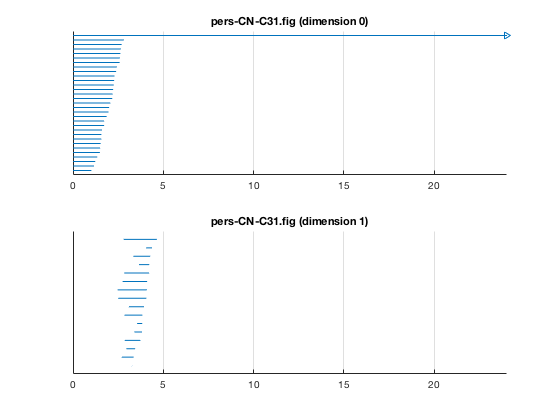}} \\
\caption{Dowker persistence diagrams for sub-networks obtained from the 2007 Asia machinery production network.}
    \label{fig: 2007 dendrogram 3}
\end{figure*}

We note that from the result illustrated in Section 4.1, JPN-C34 also has the highest quasi-centrality value in the 2007 Asia machinery production network. Since the quasi-centrality measures the propagating influence of JPN-C34 and the network hierarchy reflects JPN-C34's impact in the overall topological structure of the network, we conclude that perturbations originating from JPN-C34 not only effectively propagate throughout the entire network, but also changes the network topology significantly. 

Similarly, when $t < 0.33$, $\text{CN-C31} \not \in \mathcal{B}_{\mathcal{H}_G}(\text{STANDARD},t)$, implying that CN-C31 also has a significant impact in the network topology. Indeed, the 1-dimensional barcodes of $\Dgm(f(\gamma(G),\text{CN-C31}))$ (Figure \ref{fig: 2007 dendrogram 3}(d)) also differ significantly from the 1-dimensional barcodes of $\Dgm(\gamma(G))$ (Figure \ref{fig: 2007 dendrogram 3}(a)). 
Hence similar to JPN-C34, we conclude that perturbations originating from CN-C31 not only effectively propagate throughout the entire network, but also changes the network topology significantly.

On the contrary, although VNM-C31 has zero quasi-centrality value from the result in Section 4.1, when $t < 0.43$, $\text{VNM-C31} \not \in \mathcal{B}_{\mathcal{H}_G}(\text{STANDARD},t)$, implying that VNM-C31 has a significant impact in the network topology. Indeed, the 1-dimensional barcodes of $\Dgm(f(\gamma(G),\text{VNM-C31}))$ (Figure \ref{fig: 2007 dendrogram 3}(c)) exhibit considerable differences compared to the one-dimensional barcodes of $\Dgm(\gamma(G))$ (Figure \ref{fig: 2007 dendrogram 3}(a)). However, the 0-dimensional persistence diagrams do not reflect significant differences. Since the quasi-centrality is measured based on how much VNM-C31 contributes to the overall connectivity of the network and is measured in terms of connected components (0-dimensional homological features), we conclude that although perturbations originating from VNM-C31 cannot propagate effectively throughout the entire network, it can potentially cause changes in the network topology.

 \begin{figure*}[ht!]
\subfloat[ $\Dgm(f(\gamma(G),\text{THA-C31}))$. \label{fig:PKR}]{%
      \includegraphics[ width=0.5\textwidth]{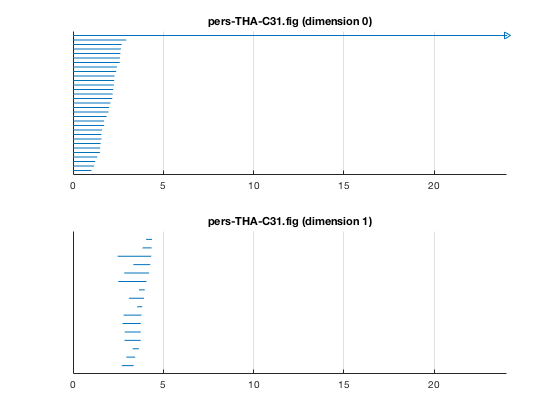}}
\hspace{\fill}
   \subfloat[$\Dgm(f(\gamma(G),\text{MYS-C29}))$. \label{fig:tie5}]{%
      \includegraphics[ width=0.5\textwidth]{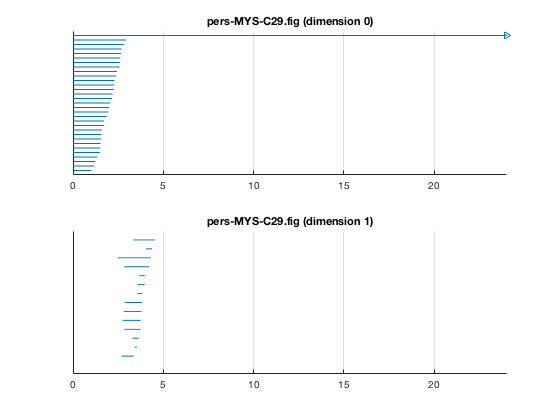}}
\hspace{\fill}
   \subfloat[$\Dgm(f(\gamma(G),\text{VNM-C30}))$. \label{fig:tie5}]{%
      \includegraphics[ width=0.5\textwidth]{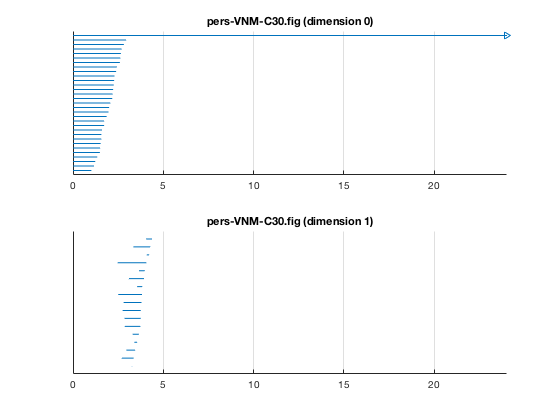}} 
\hspace{\fill}
   \subfloat[$\Dgm(f(\gamma(G),\text{VNM-C29}))$. \label{fig:tie5}]{%
      \includegraphics[ width=0.5\textwidth]{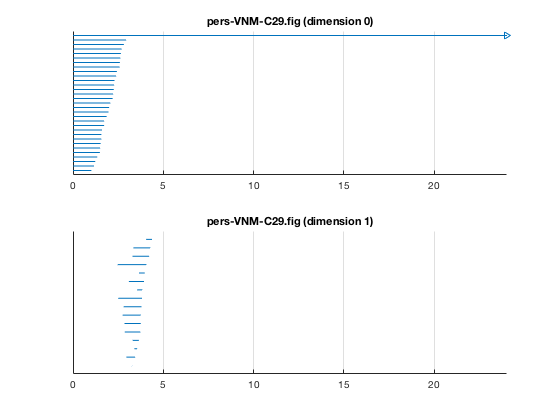}} \\
\caption{Dowker persistence diagrams for sub-networks obtained from the 2007 Asia machinery production network.}
    \label{fig: 2007 dendrogram 4}
\end{figure*}

In fact, the hierarchical dendrogram produced using the method in Definition \ref{def: hierarc} not only illustrates the hierarchy of the nodes based on their topological impact in the network, but also allows us to identify which collections of nodes have \emph{similar} impacts in the network. In particular, the nodes that first belong to $\mathcal{B}_{\mathcal{H}_G}(\text{STANDARD},t)$ at approximately the same value of $t$ are considered to have similar topological influence in the network. 

For example, it is clear from Figure \ref{fig: 2007 asia gamma dendrogram} that the nodes  THA-C31, MYS-C29, VNM-C30, VNM-C29 all first join $\mathcal{B}_{\mathcal{H}_G}(\text{STANDARD},t)$ at \emph{close} values of $t$ ($0.26 <t< 0.28$), implying that this set of nodes have similar topological influence in the network. 
Indeed, a close resemblance between the Dowker persistence diagrams corresponding to this set of nodes can be observed in Figure \ref{fig: 2007 dendrogram 4}.

\subsection{Hierarchical clustering for the 2011 Asia machinery production network}
\begin{figure}
    \centering
    \includegraphics[width=14cm]{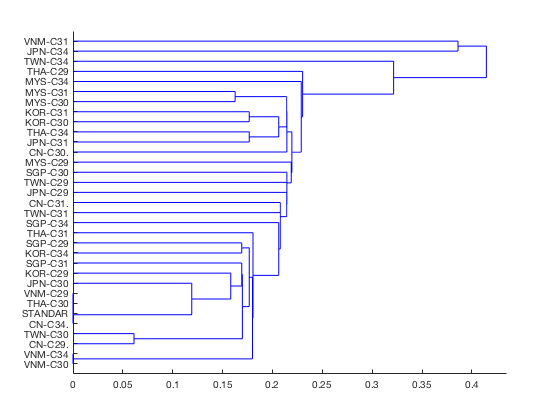}
    \caption{Hierarchichal dendrogram associated to the 2011 Asia machinery production network.}
    \label{fig: 2011 asia dendrogram}
\end{figure}

Now we analyze the hierarchical dendrogram associated to the 2011 Asia machinery production network obtained from Definition \ref{def: hierarc}. 
We hope that comparing the 2011 network hierarchical dendrogram (Figure  \ref{fig: 2011 asia dendrogram}) with the 2007 network hierarchical dendrogram (Figure  \ref{fig: 2007 asia gamma dendrogram}) will afford insights into how the country-industry pairs' influence in the topological structure of the network changed in the wake of the financial crisis. 

In contrast to the 2007 hierarchical dendrogram where $\text{JPN-C34} \not \in \mathcal{B}_{\mathcal{H}_G}(\text{STANDARD},t)$ when $t<0.48$ while $ \mathcal{B}_{\mathcal{H}_G}(\text{STANDARD},0.42) = \{S_G \setminus  \{\text{JPN-C34}\} \} $, the 2011 hierarchical dendrogram  (Figure  \ref{fig: 2011 asia dendrogram}) reflects that the absolute power of the node JPN-C34 in the network has significantly weakened; it now shares the influential role with another node: VNM-C31. The position change of JPN-C34 in the hierarchy further reflects that over the course of 2007 to 2011, the Japanese industry has stagnated \cite{doi:10.1080/10168737.2016.1148398} both in terms of its topological impacts in the structure of the network (reflected by the hierarchy) and its propagating influence through the entire network (reflected by the quasi-centrality).  

The Dowker persistence diagram of the original network $\gamma(G)$ is illustrated in Figure \ref{fig: 2011 dendrogram 2}(a) and the Dowker persistence diagrams obtained by deleting the node JPN-C34  (VNM-C31, resp.) from $\gamma(G)$ are illustrated in Figure \ref{fig: 2011 dendrogram 2}(b) ((c), resp.) respectively. We see that both JPN-C34 and VNM-C31 had cause big changes in the network topology in terms of the 1-cycles. Indeed, because JPN-C34 and VNM-C31 are "well-connected" to other nodes, deleting these nodes have a high potential disrupt the 1-dimensional homological features.

\begin{figure*}[ht!]
\subfloat[$\Dgm(\gamma(G))$. \label{fig:PKR}]{
      \includegraphics[ width=0.5 \textwidth]{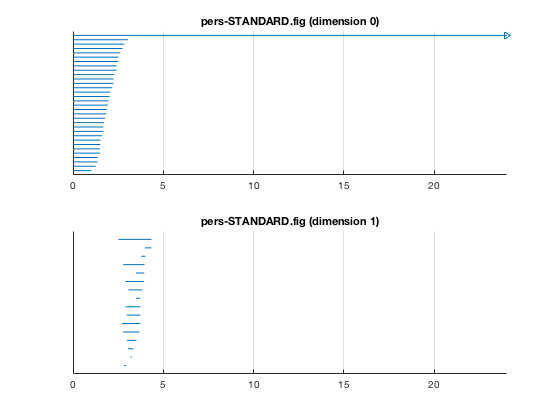}}
\hspace{\fill}
\subfloat[ $\Dgm(f(\gamma(G),\text{JPN-C34}))$. \label{fig:PKR}]{
      \includegraphics[ width=0.5\textwidth]{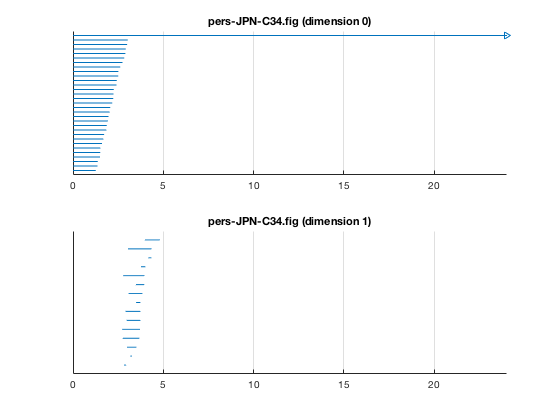}}
\hspace{\fill}
   \subfloat[$\Dgm(f(\gamma(G),\text{VNM-C31}))$. \label{fig:tie5}]{
      \includegraphics[ width=0.5\textwidth]{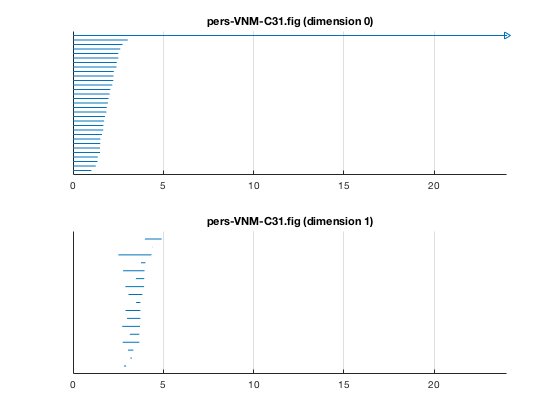}}
\hspace{\fill}
   \subfloat[$\Dgm(f(\gamma(G),\text{CN-C30}))$. \label{fig:tie5}]{
      \includegraphics[ width=0.5\textwidth]{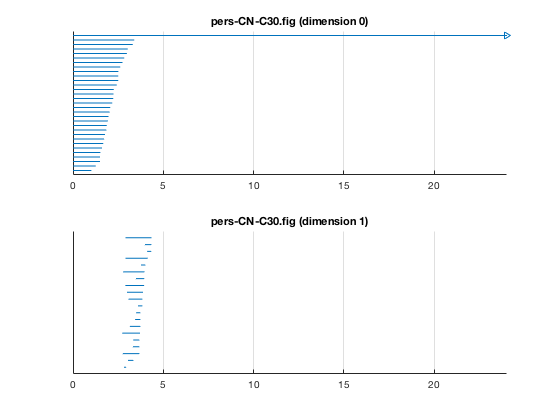}} \\
\caption{Dowker persistence diagrams for sub-networks obtained from the 2011 Asia machinery production network.}
    \label{fig: 2011 dendrogram 2}
\end{figure*}

Although CN-C30 had the largest quasi-centrality value as computed in Section 4.2, and that the Dowker persistence diagram of $f(\gamma(G),\text{CN-C30})$ illustrated in \ref{fig: 2011 dendrogram 2}(c) also reflects significant changes in the 0-dimensional barcodes compared to the Dowker persistence diagram of $\gamma(G)$, its impact on the 1-dimensional homological features is notably less than JPN-C34 and VNM-C31. 

We note that the hierarchical clustering method incorporates higher dimensional topology in the network while the quasi-centrality gives a direct measure of the extent of propagating influence of a node in terms of changes in the 0-dimensional homological features.

\subsection{Summary}

We used the bottleneck distance between persistence diagrams obtained from the Dowker sink complex, to cluster nodes by their impact on the topology of the network, with a high emphasis on higher order homological features. 
We showed that our hierarchical clustering dendrogram reflects both the extent to which each node changes the topology of the network and which collections of persistence diagrams have similar topological effects. 
 
We note the key distinction between computing quasi-centrality on nodes and hierarchical clustering: while the quasi-centrality gives a direct measure of the node's propagating influence by recording how frequently the node serves as a bridge between path components, i.e., 0-dimensional homological features, the hierarchical clustering method directly compares the Dowker persistence diagrams, which includes higher dimensional homological features. 
While the deletion of a node can potentially have a smaller impact on the 0-dimensional persistence barcodes, it might disrupt 1-cycles or higher dimensional cycles, illustrating its impact in the overall network topology.

\section{Discussion}
  
\paragraph{Summary}In this paper, we apply persistent homology to develop a network centrality measure \emph{quasi-centrality} that captures propagating influences in directed networks. 
We computed the quasi-centrality on the Asia machinery production network and demonstrated that our centrality measure accurately reflects the propagating influences of perturbations originating from individual nodes.

Furthermore, we introduced a method that allows us to express the hierarchy of nodes based on their impact in the topology of a directed network. By incorporating hierarchical clustering and the bottleneck distance, we are able to not only detect which nodes play a significant role in the topology of the network, but also determine which nodes change the topology of the network in a similar fashion.

\paragraph{Future directions} We expect that many other properties in networks (directed or undirected) admit characterizations using tools in TDA: one direction for future research is to define other measures in network analysis using TDA, such as density, robustness, efficiency, connectivity, etc. We expect that using TDA to investigate network properties allows us to bring a different perspective to examine networks.  

Another open question is how to relate higher dimensional homological features in the simplicial complexes that arises from networks to real-world phenomena, such as trade flow, embargo, value and supply chains, patterns of trade, biological cycles, etc. 

Another direction for future research is to apply the quasi-centrality measure on other directed networks. We expect that the quasi-centrality can not only be useful in assessing the influence of individual components within industries or trade networks, but also in other highly clustered complex directed networks, such as biological networks, air flight networks, etc.

We believe that many network structures such as trade flows, supply chains, and biological cycles can be better understood from a topological point of view by employing tools in TDA. Just as TDA has been successfully applied to detect intrinsic shapes in complex data sets, we believe that TDA can also find interesting structures in network topology and 
detect structural properties in networks, allowing us to better understand these networks, and gain more insights into the complex systems they model.

\paragraph{Limitations} We note that although the quasi-centrality measure and the hierarchical clustering function we defined in this paper can be useful for detecting nodes' propagating influences and topological impacts in a given directed network, computations and dendrograms generated are highly computationally costly as the number of nodes in the network exceeds 70-80 nodes.

\section{Acknowledgements}
I would like to sincerely thank my mentor Lucy Yang of Harvard University for providing thorough guidance and invaluable advice. I would like to thank the MIT-PRIMES program for providing me this invaluable opportunity and resources to work on this project. I would like to thank Professor Memoli of the Ohio State University and Dr. Chowdury of Stanford University for their generous help and discussions. I would also like to thank Dr. Tanya Khovanova and Dr. Kent Vashaw of MIT for carefully proofreading the paper and offering helpful suggestions.



\printbibliography

@article{DBLP:journals/ficn/Dabaghian20,
  author    = {Yuri Dabaghian},
  title     = {From Topological Analyses to Functional Modeling: The Case of Hippocampus},
  journal   = {Frontiers Comput. Neurosci.},
  volume    = {14},
  pages     = {593166},
  year      = {2020},
  url       = {https://doi.org/10.3389/fncom.2020.593166},
  doi       = {10.3389/fncom.2020.593166},
  bibsource = {dblp computer science bibliography, https://dblp.org}
}

@article{Topaz_2015,
   title={Topological Data Analysis of Biological Aggregation Models},
   volume={10},
   ISSN={1932-6203},
   url={http://dx.doi.org/10.1371/journal.pone.0126383},
   DOI={10.1371/journal.pone.0126383},
   number={5},
   journal={PLOS ONE},
   publisher={Public Library of Science (PLoS)},
   author={Topaz, Chad M. and Ziegelmeier, Lori and Halverson, Tom},
   editor={Ermentrout, BardEditor},
   year={2015},
   month={May},
   pages={e0126383}
}

@article{DBLP:journals/corr/ChepushtanovaEH15,
  author    = {Sofya Chepushtanova and
               Tegan Emerson and
               Eric M. Hanson and
               Michael Kirby and
               Francis C. Motta and
               Rachel Neville and
               Chris Peterson and
               Patrick D. Shipman and
               Lori Ziegelmeier},
  title     = {Persistence Images: An Alternative Persistent Homology Representation},
  journal   = {CoRR},
  volume    = {abs/1507.06217},
  year      = {2015},
  url       = {http://arxiv.org/abs/1507.06217},
  archivePrefix = {arXiv},
  eprint    = {1507.06217},
  bibsource = {dblp computer science bibliography, https://dblp.org}
}

@book{wright_2015, place={Amsterdam}, title={International encyclopedia of the social \& behavioral sciences}, publisher={Elsevier}, author={Wright, James D.}, year={2015}}

@article{doi:10.1137/S003614450342480,
author = {Newman, M. E. J.},
title = {The Structure and Function of Complex Networks},
journal = {SIAM Review},
volume = {45},
number = {2},
pages = {167-256},
year = {2003},
doi = {10.1137/S003614450342480},

URL = { 
        https://doi.org/10.1137/S003614450342480
    
},
eprint = { 
        https://doi.org/10.1137/S003614450342480
    
}

}

@article{Costa_2011,
   title={Analyzing and modeling real-world phenomena with complex networks: a survey of applications},
   volume={60},
   ISSN={1460-6976},
   url={http://dx.doi.org/10.1080/00018732.2011.572452},
   DOI={10.1080/00018732.2011.572452},
   number={3},
   journal={Advances in Physics},
   publisher={Informa UK Limited},
   author={Costa, Luciano da Fontoura and Oliveira, Osvaldo N. and Travieso, Gonzalo and Rodrigues, Francisco Aparecido and Villas Boas, Paulino Ribeiro and Antiqueira, Lucas and Viana, Matheus Palhares and Correa Rocha, Luis Enrique},
   year={2011},
   month={Jun},
   pages={329–412}
}

@misc{şimşek2021geometry,
      title={Geometry of the Loss Landscape in Overparameterized Neural Networks: Symmetries and Invariances}, 
      author={Berfin Şimşek and François Ged and Arthur Jacot and Francesco Spadaro and Clément Hongler and Wulfram Gerstner and Johanni Brea},
      year={2021},
      eprint={2105.12221},
      archivePrefix={arXiv},
      primaryClass={cs.LG}
}

@article{Leifer_2020,
   title={Circuits with broken fibration symmetries perform core logic computations in biological networks},
   volume={16},
   ISSN={1553-7358},
   url={http://dx.doi.org/10.1371/journal.pcbi.1007776},
   DOI={10.1371/journal.pcbi.1007776},
   number={6},
   journal={PLOS Computational Biology},
   publisher={Public Library of Science (PLoS)},
   author={Leifer, Ian and Morone, Flaviano and Reis, Saulo D. S. and Andrade, José S. and Sigman, Mariano and Makse, Hernán A.},
   editor={Pascual, MercedesEditor},
   year={2020},
   month={Jun},
   pages={e1007776}
}

@article{doi:10.1177/016555150202800601,
author = {Evelien Otte and Ronald Rousseau},
title ={Social network analysis: a powerful strategy, also for the information sciences},
journal = {Journal of Information Science},
volume = {28},
number = {6},
pages = {441-453},
year = {2002},
doi = {10.1177/016555150202800601},

URL = { 
        https://doi.org/10.1177/016555150202800601
    
},
eprint = { 
        https://doi.org/10.1177/016555150202800601
    
}
}

@article{PhysRevE.68.015101,
  title = {Topology of the world trade web},
  author = {Serrano, Ma \'Angeles and Bogu\~n\'a, Mari\'an},
  journal = {Phys. Rev. E},
  volume = {68},
  issue = {1},
  pages = {015101},
  numpages = {4},
  year = {2003},
  month = {Jul},
  publisher = {American Physical Society},
  doi = {10.1103/PhysRevE.68.015101},
  url = {https://link.aps.org/doi/10.1103/PhysRevE.68.015101}
}

@article{brockmann_helbing_2013, title={The Hidden Geometry of Complex, Network-Driven Contagion Phenomena}, volume={342}, DOI={10.1126/science.1245200}, number={6164}, journal={Science}, author={Brockmann, D. and Helbing, D.}, year={2013}, pages={1337–1342}}

@book{Munkers84,
  author = {Munkres, James R.},
  howpublished = {Hardcover},
  isbn = {0201045869},
  publisher = {Addison Wesley Publishing Company},
  title = {{Elements of Algebraic Topology}},
  url = {http://www.worldcat.org/isbn/0201045869},
  year = 1984
}

@article{chowdhury_mémoli_2018, title={A functorial Dowker theorem and persistent homology of asymmetric networks}, volume={2}, DOI={10.1007/s41468-018-0020-6}, number={1-2}, journal={Journal of Applied and Computational Topology}, author={Chowdhury, Samir and Mémoli, Facundo}, year={2018}, pages={115–175}}

@article{rider_2005, title={Social Economics}, DOI={10.1016/b0-12-369398-5/00552-1}, journal={Encyclopedia of Social Measurement}, author={Rider, Christine}, year={2005}, pages={501–508}}

@article{Lee_2021,
   title={Betweenness centrality of teams in social networks},
   volume={31},
   ISSN={1089-7682},
   url={http://dx.doi.org/10.1063/5.0056683},
   DOI={10.1063/5.0056683},
   number={6},
   journal={Chaos: An Interdisciplinary Journal of Nonlinear Science},
   publisher={AIP Publishing},
   author={Lee, Jongshin and Lee, Yongsun and Oh, Soo Min and Kahng, B.},
   year={2021},
   month={Jun},
   pages={061108}
}

@misc{taylor2016eigenvectorbased,
      title={Eigenvector-Based Centrality Measures for Temporal Networks}, 
      author={Dane Taylor and Sean A. Myers and Aaron Clauset and Mason A. Porter and Peter J. Mucha},
      year={2016},
      eprint={1507.01266},
      archivePrefix={arXiv},
      primaryClass={physics.soc-ph}
}

@misc{delima2020estimating,
      title={Estimating the Percolation Centrality of Large Networks through Pseudo-dimension Theory}, 
      author={Alane M. de Lima and Murilo V. G. da Silva and André L. Vignatti},
      year={2020},
      eprint={1910.00494},
      archivePrefix={arXiv},
      primaryClass={cs.DS}
}

@article{Liu_2020,
   title={K-Core based Temporal Graph Convolutional Network for Dynamic Graphs},
   ISSN={2326-3865},
   url={http://dx.doi.org/10.1109/TKDE.2020.3033829},
   DOI={10.1109/tkde.2020.3033829},
   journal={IEEE Transactions on Knowledge and Data Engineering},
   publisher={Institute of Electrical and Electronics Engineers (IEEE)},
   author={Liu, Jingxin and Xu, Chang and Yin, Chang and Wu, Weiqiang and Song, You},
   year={2020},
   pages={1–1}
}

@article{doi:10.1080/10168737.2016.1148398,
author = {Ayako Obashi and Fukunari Kimura},
title = {The Role of China, Japan, and Korea in Machinery Production Networks},
journal = {International Economic Journal},
volume = {30},
number = {2},
pages = {169-190},
year  = {2016},
publisher = {Routledge},
doi = {10.1080/10168737.2016.1148398},

URL = { 
        https://doi.org/10.1080/10168737.2016.1148398
    
},
eprint = { 
        https://doi.org/10.1080/10168737.2016.1148398
    
}}

@misc{melnik2016simple,
      title={Simple and accurate analytical calculation of shortest path lengths}, 
      author={Sergey Melnik and James P. Gleeson},
      year={2016},
      eprint={1604.05521},
      archivePrefix={arXiv},
      primaryClass={physics.soc-ph}
}

@article{Masuda_2018,
   title={Clustering Coefficients for Correlation Networks},
   volume={12},
   ISSN={1662-5196},
   url={http://dx.doi.org/10.3389/fninf.2018.00007},
   DOI={10.3389/fninf.2018.00007},
   journal={Frontiers in Neuroinformatics},
   publisher={Frontiers Media SA},
   author={Masuda, Naoki and Sakaki, Michiko and Ezaki, Takahiro and Watanabe, Takamitsu},
   year={2018},
   month={Mar}
}

@misc{ficara2021correlation,
      title={Correlation analysis of node and edge centrality measures in artificial complex networks}, 
      author={Annamaria Ficara and Giacomo Fiumara and Pasquale De Meo and Antonio Liotta},
      year={2021},
      eprint={2103.05427},
      archivePrefix={arXiv},
      primaryClass={cs.SI}
}

@TECHREPORT{RePEc:cpr:ceprdp:10631,
title = {Teams, Organization and Education Outcomes: Evidence from a field experiment in Bangladesh},
author = {Hahn, Youjin and Islam, Asad and Patacchini, Eleonora and Zenou, Yves},
year = {2015},
institution = {C.E.P.R. Discussion Papers},
type = {CEPR Discussion Papers},
number = {10631},
url = {https://EconPapers.repec.org/RePEc:cpr:ceprdp:10631}
}

@article{GOFMAN2017113,
title = {Efficiency and stability of a financial architecture with too-interconnected-to-fail institutions},
journal = {Journal of Financial Economics},
volume = {124},
number = {1},
pages = {113-146},
year = {2017},
issn = {0304-405X},
doi = {https://doi.org/10.1016/j.jfineco.2016.12.009},
url = {https://www.sciencedirect.com/science/article/pii/S0304405X16302471},
author = {Michael Gofman}
}

@article{Serrano_2007,
   title={Patterns of dominant flows in the world trade web},
   volume={2},
   ISSN={1860-7128},
   url={http://dx.doi.org/10.1007/s11403-007-0026-y},
   DOI={10.1007/s11403-007-0026-y},
   number={2},
   journal={Journal of Economic Interaction and Coordination},
   publisher={Springer Science and Business Media LLC},
   author={Serrano, M. Ángeles and Boguñá, Marián and Vespignani, Alessandro},
   year={2007},
   month={Oct},
   pages={111–124}
}

@misc{agami2020comparison,
      title={Comparison of Persistence Diagrams}, 
      author={Sarit Agami},
      year={2020},
      eprint={2003.01352},
      archivePrefix={arXiv},
      primaryClass={stat.AP}
}

@article{korniyenko_pinat_dew_2017, title={Assessing the Fragility of Global Trade: The Impact of Localized Supply Shocks Using Network Analysis}, volume={17}, DOI={10.5089/9781475578515.001}, number={30}, journal={IMF Working Papers}, author={Korniyenko, Yevgeniya and Pinat, Magali and Dew, Brian}, year={2017}, pages={1}}

@article{econometrica_2012, title={The Network Origins of Aggregate Fluctuations}, volume={80}, DOI={10.3982/ecta9623}, number={5}, journal={Econometrica}, year={2012}, pages={1977–2016}}

@article{Iannelli_2017,
   title={Effective distances for epidemics spreading on complex networks},
   volume={95},
   ISSN={2470-0053},
   url={http://dx.doi.org/10.1103/PhysRevE.95.012313},
   DOI={10.1103/physreve.95.012313},
   number={1},
   journal={Physical Review E},
   publisher={American Physical Society (APS)},
   author={Iannelli, Flavio and Koher, Andreas and Brockmann, Dirk and Hövel, Philipp and Sokolov, Igor M.},
   year={2017},
   month={Jan}
}

@article{otter_porter_tillmann_grindrod_harrington_2017, title={A roadmap for the computation of persistent homology}, volume={6}, DOI={10.1140/epjds/s13688-017-0109-5}, number={1}, journal={EPJ Data Science}, author={Otter, Nina and Porter, Mason A and Tillmann, Ulrike and Grindrod, Peter and Harrington, Heather A}, year={2017}}

@book{/content/publication/9789264237087-en,
   author = "OECD",
   title = "Green Growth in Bangkok, Thailand",
   year = "2015",
   pages = 148,
   url = "https://www.oecd-ilibrary.org/content/publication/9789264237087-en",
   doi = "https://doi.org/https://doi.org/10.1787/9789264237087-en" 
}

@book{hoon_2022, place={New York}, title={The Singapore economy dynamism and inclusion}, publisher={Routledge}, author={Hoon, Hian Teck}, year={2022}}

@book{kulish_2002, place={Dordrecht}, title={Hierarchical methods}, publisher={Kluwer Academic Publishers}, author={Kulish, V. V.}, year={2002}}

@book{rabadan.raul.blumberg.20,
  title     = "Topological data analysis for genomics and evolution: topology in biology",
  author    = "Rabadan Raul and Blumberg, Andrew J",
  year      = 2020,
  publisher = "Cambridge university press",
  address   = "Cambridge"
}

@online{OECD2021,
  author = {OECD},
  year = {2021},
  title = {OECD Inter-Country Input-Output (ICIO) Tables},
  journal = {data retrieved from OECD Inter-Country Input-Output (ICIO) Tables}, 
  url = {https://www.oecd.org/sti/ind/inter-country-input-output-tables.htm},
  urldate = {2021-08-30}
}

@Online{OEC.2021,
 author = {OEC},
 year = {2021},
 title = {Life expectancy},
 journal = {OEC},
 url = {https://oec.world/},
 urldate = {2021-08-30}
}

@INPROCEEDINGS{8862232,
  author={Vijaya and Sharma, Shweta and Batra, Neha},
  booktitle={2019 International Conference on Machine Learning, Big Data, Cloud and Parallel Computing (COMITCon)}, 
  title={Comparative Study of Single Linkage, Complete Linkage, and Ward Method of Agglomerative Clustering}, 
  year={2019},
  volume={},
  number={},
  pages={568-573},
  doi={10.1109/COMITCon.2019.8862232}}
\end{document}